\documentclass[letterpaper,twocolumn,10pt]{article}

\usepackage{usenix2019_v3,epsfig,endnotes}
\usepackage{multirow}
\usepackage{graphicx}
\usepackage{subfig}
\usepackage{grffile}
\usepackage{thmtools, thm-restate}

\newcounter{subcopyrightbox@save}
\usepackage{caption}
\usepackage{color, url}
\usepackage{xspace} 
\usepackage{mathrsfs}
\usepackage{amssymb}
\usepackage{amsmath}
\usepackage{amsthm}
\usepackage{epstopdf}
\usepackage{balance}
\usepackage{bm}

\usepackage{algorithm}
\usepackage{algorithmic}

\usepackage{lipsum}
\usepackage{xcolor}

\newcommand\hl[1]{%
  \bgroup
  \hskip0pt\color{red!80!black}%
  #1%
  \egroup
}

\newtheorem{theorem}{Theorem}

\newcommand{\myparatight}[1]{\smallskip\noindent{\bf {#1}:}~}
\newenvironment{packeditemize}{\begin{list}{$\bullet$}{\setlength{\itemsep}{0.2pt}\addtolength{\labelwidth}{10pt}\setlength{\leftmargin}{\labelwidth}\setlength{\listparindent}{\parindent}\setlength{\parsep}{1pt}\setlength{\topsep}{0pt}}}{\end{list}}

\newcommand{\neil}[1]{{#1}}

\newcommand{\alan}[1]{{#1}}

\graphicspath{ {figs/} }

\title{\Large \bf Local Model Poisoning Attacks to Byzantine-Robust Federated Learning}

\author{
{\rm Minghong Fang\thanks{Equal contribution. Minghong Fang performed this research when he was under the supervision of Neil Zhenqiang Gong.}$\:\;^{1}$, Xiaoyu Cao\textcolor{green!80!black}{\footnotemark[1]}${\:\;^{2}}$, Jinyuan Jia$^2$, Neil Zhenqiang Gong$^2$} \\
$^1$ECE Department, The Ohio State University, $^2$ECE Department, Duke University\\
$^1$fang.841@osu.edu, $^2$\{xiaoyu.cao, jinyuan.jia, neil.gong\}@duke.edu}

\begin{document}
\maketitle

\thispagestyle{headings}
\markright{\hfill To appear in the 29th Usenix Security Symposium, August 2020, Boston, MA\hfill}


\begin{abstract} 
In federated learning, multiple client devices jointly learn a machine learning model: each client device maintains a local model for its local training dataset, while a master device  maintains a global model via aggregating the local models from the client devices. The machine learning community recently proposed several federated learning methods that were claimed to be robust against Byzantine failures (e.g., system failures, adversarial manipulations) of certain client devices. 
\neil{In this work, we perform the first systematic study on \emph{local model poisoning attacks} to federated learning. We assume an attacker has compromised some client devices, and the attacker manipulates the local model parameters on the compromised client devices during the learning process such that the global model has a large testing error rate. We formulate our attacks as optimization problems and apply our attacks to four recent Byzantine-robust federated learning methods. Our empirical results on four real-world datasets show that our attacks can substantially increase the error rates of the models learnt by the federated learning methods that were claimed to be robust against Byzantine failures of some client devices. 
We generalize two defenses for data poisoning attacks to defend against our local model poisoning attacks. 
Our evaluation results show that one defense can effectively defend against our attacks in some cases, but the defenses are not effective enough in other cases, highlighting the need for new defenses against our local model poisoning attacks to federated learning.} 
\vspace{-2mm}
\end{abstract}


\section{Introduction}
\vspace{-1mm}
\label{sec:intro}

\myparatight{Byzantine-robust federated learning} In \emph{federated learning} (also known as \emph{collaborative learning})~\cite{Konen16,McMahan17}, the training dataset is decentralized among multiple client devices (e.g., desktops, mobile phones, IoT devices), which could belong to different users or organizations. These users/organizations do not want to share their local training datasets, but still desire to jointly learn a model. For instance, multiple hospitals may desire to learn a healthcare model without sharing their  sensitive data to each other. 
Each client device (called \emph{worker device}) maintains a \emph{local model}  for its local training dataset.  Moreover, the service provider has a \emph{master device} (e.g., cloud server), which maintains a \emph{global model}. Roughly speaking, federated learning repeatedly performs three steps: the master device sends the current global model to worker devices; worker devices update their local models using their local training datasets and the global model, and send the local models to the master device; and the master device computes a new global model via aggregating the local models according to a certain \emph{aggregation rule}. 

For instance, the \emph{mean} aggregation rule that takes the average of the local model parameters as the global model is widely used under non-adversarial settings. However, the global model can be arbitrarily  manipulated for mean  even if just one worker device is compromised~\cite{Blanchard17,Yin18}. Therefore, the machine learning community recently proposed multiple aggregation rules (e.g., Krum~\cite{Blanchard17}, Bulyan~\cite{Mhamdi18}, trimmed mean~\cite{Yin18}, and median~\cite{Yin18}), which aimed to be robust against Byzantine failures of certain worker devices.

\begin{figure}[!t]
\vspace{-5mm}
\centering
{\includegraphics[width=0.3 \textwidth]{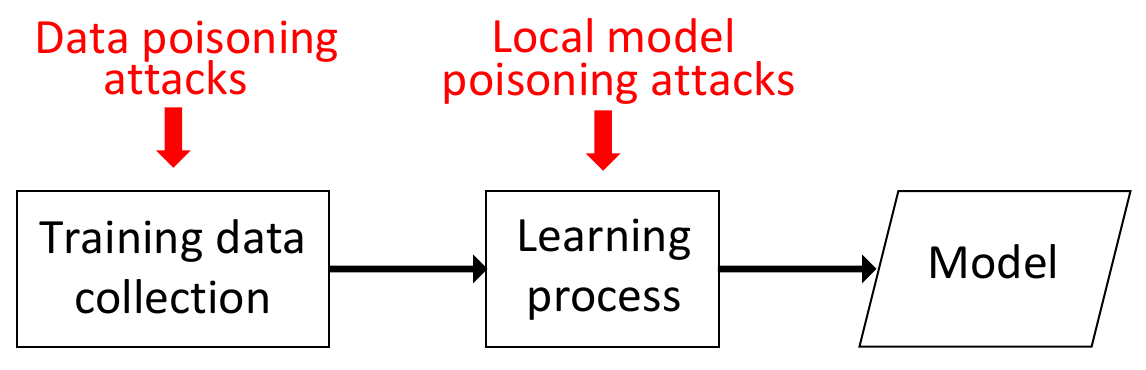}}
\caption{\emph{Data} vs. \emph{local model} poisoning attacks.}
\vspace{-4mm}
\label{flowchart}
\end{figure}

\neil{\myparatight{Existing data poisoning attacks are insufficient} 
We consider attacks that aim to manipulate the \emph{training phase} of machine learning such that the learnt model (we consider the model to be a classifier) has a high testing error rate indiscriminately for testing examples, which makes the model unusable and eventually leads to denial-of-service attacks. Figure~\ref{flowchart} shows the training phase, which includes two components, i.e., \emph{training dataset collection} and \emph{learning process}. The training dataset collection component is to collect a training dataset, while the learning process component produces a model from a given training dataset. Existing  attacks mainly inject malicious data into the training dataset before the learning process starts, while the learning process is assumed to maintain integrity. Therefore, these attacks are often called \emph{data poisoning attacks}~\cite{rubinstein2009antidote,biggio2012poisoning,xiao2015feature,poisoningattackRecSys16,Jagielski18,Suciu18}. {In federated learning,  an attacker could only inject the malicious data into the worker devices that are under the attacker's control.} As a result, these data poisoning attacks have limited success to attack Byzantine-robust federated learning (see our experimental results in Section~\ref{comparison-data}).}

\myparatight{Our work} We perform the first study on \emph{local model poisoning attacks} to Byzantine-robust federated learning. \neil{Existing studies~\cite{Blanchard17,Yin18} only showed local model poisoning attacks to federated learning with the non-robust mean aggregation rule.}

{\bf Threat model.} Unlike existing data poisoning attacks that compromise the integrity of training dataset collection, we aim to compromise the integrity of the learning process in the training phase (see Figure~\ref{flowchart}). 
We assume the attacker has control of some worker devices and manipulates the local model parameters sent from these devices to the master device during the learning process. The attacker may or may not know the aggregation rule used by the master device. To contrast with data poisoning attacks, we call our attacks {local model poisoning attacks} as they directly manipulate the local model parameters.

{\bf Local model poisoning attacks.} A key challenge of local model poisoning attacks is how to craft the local models sent from the compromised worker devices to the master device. To address this challenge, we formulate crafting local models as solving an optimization problem in each iteration of federated learning. Specifically, the master device could compute a global model in an iteration if there are no attacks, which we call \emph{before-attack} global model. Our goal is to craft the local models on the compromised worker devices such that the global model deviates the most towards the inverse of the direction along which the before-attack global model would change. Our intuition is that the deviations accumulated over multiple iterations would make the learnt global model differ from the before-attack one significantly. 
We apply our attacks to four recent Byzantine-robust federated learning methods including Krum, Bulyan, trimmed mean, and median. 

Our evaluation results on the MNIST, Fashion-MNIST,   CH-MNIST, and Breast Cancer Wisconsin (Diagnostic) datasets 
show that our attacks can substantially increase the error rates of the global models under various settings of federated learning. For instance, when learning a deep neural network classifier for  MNIST  using Krum, our attack can increase the error rate from 0.11 to 0.75. Moreover, we compare with data poisoning attacks including \emph{label flipping attacks} and \emph{back-gradient optimization based attacks}~\cite{munoz2017towards} (state-of-the-art untargeted data poisoning attacks for multi-class classifiers), which poison the local training datasets on the compromised worker devices. We find that these data poisoning attacks have limited success to attack the Byzantine-robust federated learning methods.

{\bf Defenses.} Existing defenses against data poisoning attacks essentially aim to sanitize the training dataset. 
One category of defenses~\cite{Cretu08,barreno2010security,Suciu18,Tran18} detects malicious data based on their negative impact on the error rate of the learnt model. For instance, \emph{Reject on Negative Impact (RONI)}~\cite{barreno2010security} measures the impact of each training example on the error rate of the learnt model and removes the training examples that have large negative impact. Another category of defenses~\cite{Feng14,Liu17AiSec,Jagielski18} leverages new loss functions, solving which detects malicious data and learns a model simultaneously. For instance, Jagielski et al.~\cite{Jagielski18}  proposed TRIM, which aims to jointly find a subset of training dataset with a given size and model parameters that minimize the loss function. The training examples that are not in the selected subset are treated as malicious data. However, these defenses are not directly applicable for our local model poisoning attacks because our attacks do not inject malicious data into the training dataset.

To address the challenge, we generalize RONI and TRIM to defend against our {local model poisoning attacks}. Both defenses remove the local models that are potentially malicious before computing the global model using a Byzantine-robust aggregation rule in each iteration. One defense removes the local models that have large negative impact on the error rate of the global model (inspired by RONI that removes training examples that have large negative impact on the error rate of the model), while the other defense removes the local models that result in large loss (inspired by TRIM that removes the training examples that have large negative impact on the loss), where the error rate and loss are evaluated on a validation dataset. We call the two defenses \emph{Error Rate based Rejection (ERR)} and \emph{Loss Function based Rejection (LFR)}, respectively. \neil{Moreover, we combine ERR and LFR, i.e., we remove the local models that are removed by either  ERR or LFR.  Our empirical evaluation results show that LFR outperforms ERR; and the combined defense is comparable to LFR in most cases.} Moreover, LFR can defend against our attacks in certain cases, but LFR is not effective enough in other cases. For instance, LFR can effectively defend against our attacks that craft local models based on the trimmed mean aggregation rule, but LFR is not effective against our attacks that are based on the Krum aggregation rule. Our results show that we need new defense mechanisms to defend against our local model poisoning attacks.

Our key contributions can be summarized as follows:
\begin{packeditemize}
\item We perform the first  systematic study on  attacking  Byzantine-robust federated learning. 

\item {We propose \emph{local model poisoning attacks} to Byzantine-robust federated learning.} 
Our attacks  manipulate the local model parameters on compromised worker devices during the learning process. 

\item {We generalize two defenses for data poisoning attacks to defend against local model poisoning attacks. Our results show that, although one of them is effective in some cases, they have limited success in other cases}. 

\end{packeditemize}


\section{Background and Problem Formulation}
\label{sec:background}

\subsection{Federated Learning}
\label{federatedlearning}
Suppose we have $m$ worker devices and the $i$th worker device has a local training dataset $D_i$. The worker devices aim to collaboratively learn a classifier. Specifically, the model parameters  $\mathbf{w}$ of the classifier are often obtained via solving the following optimization problem: $\min_{\mathbf{w}} \sum_{i=1}^m F(\mathbf{w}, D_i)$, 
where $F(\mathbf{w}, D_i)$ is the objective function for the local training dataset on the $i$th device and characterizes how well the parameters $\mathbf{w}$ model the local training dataset on the $i$th device. Different classifiers (e.g., logistic regression, deep neural networks) use different objective functions. 
In federated learning, each worker device maintains a local model for its local training dataset. Moreover, we have a master device to maintain a global model via aggregating local models from the $m$ worker devices. 
Specifically, federated learning performs the following three steps in each iteration:

{\bf Step I.} The master device sends the current global model parameters to all worker devices. 

{\bf Step II.} The worker devices update their local model parameters using the current global model parameters and their local training datasets in parallel. In particular, the $i$th worker device essentially aims to solve the optimization problem  $\min_{\mathbf{w}_i} F(\mathbf{w}_i, D_i)$ with the global model parameters $\mathbf{w}$ as an initialization of the local model parameters $\mathbf{w}_i$.  A worker device could use any method to solve the optimization problem, though \emph{stochastic gradient descent} is the most popular one. Specifically, the $i$th worker device updates its local model parameters $\mathbf{w}_i$ as $\mathbf{w}_i=\mathbf{w} - \alpha\cdot \frac{\partial F(\mathbf{w}, B_i)}{\partial \mathbf{w}}$, 
where $\alpha$ is the learning rate and $B_i$ is a randomly sampled batch from the local training dataset $D_i$. Note that a worker device could apply stochastic gradient descent multiple rounds to update its local model. 
After updating the local models, the worker devices send them to the master device. 

{\bf Step III.} The master device {aggregates} the local models from the worker devices to obtain a new global model according to a certain \emph{aggregation rule}. Formally, we have $ \mathbf{w}  = \mathcal{A}(\mathbf{w}_1, \mathbf{w}_2, \cdots, \mathbf{w}_m)$.

The master device could also randomly pick a subset of worker devices and send the global model to them; the picked worker devices update their local models and send them to the master device; and the master device aggregates the local models to obtain the new global model~\cite{McMahan17}. \neil{We note that, for the aggregation rules we study in this paper, sending local models to the master device is equivalent to sending gradients to the master device, who aggregates the gradients and uses them to update the global model.}

\subsection{Byzantine-robust Aggregation Rules}

A naive aggregation rule is to average the local model parameters as the global model parameters.  
This \emph{mean}  aggregation rule is widely used under non-adversarial settings~\cite{Dean12,Konen16,McMahan17}. However,   mean  is not robust under adversarial settings. In particular, an attacker can manipulate the global model parameters arbitrarily for this mean aggregation rule when compromising only one worker device~\cite{Blanchard17,Yin18}. 
Therefore, the machine learning community has recently developed multiple aggregation rules that aim to be robust even if certain worker devices exhibit Byzantine failures. 
Next, we review several such aggregation rules.

\myparatight{Krum~\cite{Blanchard17} and Bulyan~\cite{Mhamdi18}} Krum selects one of the $m$ local models  that is similar to other models as the global model. The intuition is that even if the selected local model is from a compromised worker device, its impact may be constrained since it is similar to other local models possibly from benign worker devices. Suppose at most $c$ worker devices are compromised. For each local model $\mathbf{w}_i$, the master device computes the $m-c-2$ local models that are the closest to $\mathbf{w}_i$ with respect to Euclidean distance. Moreover, the master device computes the sum of the squared distances between $\mathbf{w}_i$ and its closest $m-c-2$ local models. Krum selects the local model with the smallest sum of squared distance as the global model. When $c < \frac{m-2}{2}$, Krum has theoretical guarantees for the convergence for certain objective functions.

Euclidean distance between two local models could be substantially influenced by a single model parameter. Therefore, Krum could be influenced by some abnormal model parameters~\cite{Mhamdi18}. To address this issue, Mhamdi et al.~\cite{Mhamdi18} proposed Bulyan, which essentially combines Krum and a variant of trimmed mean (trimmed mean will be discussed next). Specifically, Bulyan first iteratively applies Krum to select $\theta$ ($\theta \leq m-2c$) local models. 
Then, Bulyan uses a variant of trimmed mean to aggregate the $\theta$ local models. In particular, for each $j$th model parameter, Bulyan sorts the $j$th parameters of the $\theta$ local models, finds the $\gamma$ ($\gamma \leq \theta - 2c$) parameters that are the closest to the median, and computes their mean as the $j$th parameter of the global model. When $c \leq \frac{m-3}{4}$, Bulyan has theoretical guarantees for the convergence under certain assumptions of the objective function.

Since Bulyan is based on Krum, our attacks for Krum can transfer to Bulyan (see Appendix~\ref{attackBulyan}). Moreover, Bulyan is not scalable because it executes Krum many times in each iteration and Krum computes pairwise distances between local models.  
Therefore, we will focus on Krum in the paper.   

\myparatight{Trimmed mean~\cite{Yin18}}  This aggregation rule aggregates each model parameter independently. Specifically, for each $j$th model parameter, the master device sorts the $j$th parameters of the $m$ local models, i.e., $w_{1j}, w_{2j}, \cdots, w_{mj}$, where $w_{ij}$ is the $j$th parameter of the $i$th local model, removes the largest and smallest $\beta$ of them, and computes the mean of the remaining $m-2\beta$ parameters as the $j$th parameter of the global model. Suppose at most $c$ worker devices are compromised. This trimmed mean aggregation rule achieves \emph{order-optimal} error rate when ${c} \leq \beta < \frac{m}{2}$ and the objective function to be minimized is strongly convex. Specifically, the order-optimal error rate is $\tilde{O}(\frac{c}{m\sqrt{n}} + \frac{1}{\sqrt{mn}})$,\footnote{$\tilde{O}$ is a variant of the $O$ notation, which ignores the logarithmic terms.\label{fn:1}} where $n$ is the number of training data points on a worker device (worker devices are assumed to have the same number of training data points). 

\myparatight{Median~\cite{Yin18}} In this median aggregation rule, for each  $j$th model parameter, the master device sorts the $j$th parameters of the $m$ local models and takes the median as the $j$th parameter of the global model. Note that when $m$ is an even number, median is the mean of the middle two parameters. Like the trimmed mean aggregation rule, the median aggregation rule also achieves an order-optimal error rate when the objective function is strongly convex.

\subsection{Problem Definition and Threat Model}
\label{sec:threatmodel}

\myparatight{Attacker's goal} Like many studies on poisoning attacks~\cite{rubinstein2009antidote,biggio2012poisoning,biggio2013poisoning,xiao2015feature,Jagielski18,poisoningattackRecSys16,YangRecSys17}, we consider an attacker's goal is to manipulate the learnt global model such that it has a high error rate indiscriminately for testing examples. Such attacks are known as \emph{untargeted poisoning attacks}, which make the learnt model unusable and eventually lead to denial-of-service attacks. \alan{For instance, an attacker may perform such attacks to its competitor's federated learning system.} Some studies also considered other types of poisoning attacks (e.g., \emph{targeted poisoning attacks}~\cite{Suciu18}), which we will review in Section~\ref{related}.

We note that the Byzantine-robust aggregation rules discussed above can \emph{asymptotically} bound the error rates of the learnt global model under certain assumptions of the objective functions, and some of them (i.e., trimmed mean and median) even achieve \emph{order-optimal} error rates. 
These theoretical guarantees seem to imply the difficulty of manipulating the error rates. However, the asymptotic guarantees do not precisely characterize the \emph{practical} performance of the learnt models. Specifically, the asymptotic error rates are quantified using the $\tilde{O}$ notation. The $\tilde{O}$ notation ignores any constant, e.g., $\tilde{O}(\frac{1}{\sqrt{n}})$=$\tilde{O}(\frac{100}{\sqrt{n}})$. However, such constant significantly influences a model's error rate in practice. As we will show, although these asymptotic error rates still hold for our local model poisoning attacks since they hold for Byzantine failures, our attacks can still significantly increase the testing error rates of the learnt models in practice.

\myparatight{Attacker's capability} We assume the attacker has control of $c$ worker devices. \alan{Specifically, like Sybil attacks~\cite{sybil} to distributed systems, the attacker could inject $c$ fake worker devices into the federated learning system or compromise $c$ benign worker devices.} However, we assume the number of worker devices under the attacker's control is less than 50\% (otherwise, it would be easy to manipulate the global models). We assume the attacker can arbitrarily manipulate the local models sent from these worker devices to the master device. For simplicity, we call these worker devices \emph{compromised worker devices} no matter whether they are fake devices or compromised benign ones.

\myparatight{Attacker's background knowledge} The attacker knows the code, local training datasets, and local models on the compromised worker devices. 
We characterize the attacker's background knowledge along the following two dimensions:

{\bf Aggregation rule.} 
We consider two scenarios depending on whether the attacker knows the aggregation rule or not. In particular, the attacker could know the aggregation rule in various scenarios. 
For instance, the service provider may make the aggregation rule public in order to increase transparency and trust of the federated learning system~\cite{McMahan17}.  
When the attacker does not know the aggregation rule, we will craft local model parameters for the compromised worker devices based on a certain aggregation rule. Our empirical results show that such crafted local models could also attack other aggregation rules. In particular, we observe different levels of \emph{transferability} of our local model poisoning attacks between different aggregation rules. 

{\bf Training data.}  We consider two cases (\emph{full knowledge} and \emph{partial knowledge}) depending on whether the attacker knows the local training datasets and local models on the benign worker devices. In the full knowledge scenario,  the attacker knows the local training dataset and local model on every worker device. 
We note that the full knowledge scenario has limited applicability in practice for federated learning as the training dataset is decentralized on many worker devices, and \neil{we use it to estimate the \emph{upper bound} of our attacks' threats for a given setting of federated learning}. In the partial knowledge scenario, the attacker only knows the local training datasets and local models on the compromised worker devices.

Our threat model is inspired by multiple existing studies~\cite{Papernot16,Papernot16Distillation,Jagielski18,Suciu18} on adversarial machine learning. 
For instance, Suciu et al.~\cite{Suciu18} recently proposed to characterize an attacker's background knowledge and capability for data poisoning attacks with respect to multiple dimensions such as \emph{Feature}, \emph{Algorithm}, and \emph{Instance}. Our aggregation rule and training data dimensions are essentially the Algorithm and Instance dimensions, respectively. We do not consider the Feature dimension because the attacker controls some worker devices and already knows the features in our setting. 

Some Byzantine-robust aggregation rules (e.g.,  Krum~\cite{Blanchard17} and trimmed mean~\cite{Yin18}) need to know the upper bound of the number of compromised worker devices in order to set parameters appropriately. For instance,  trimmed mean  removes the largest and smallest $\beta$ local model parameters, where $\beta$ is at least the number of compromised worker devices (otherwise trimmed mean can be easily manipulated). \alan{To calculate a lower bound for our attack's threat, we consider a hypothetical, strong service provider who knows the number of compromised worker devices and sets parameters in the aggregation rule accordingly.}


\section{Our Local Model Poisoning Attacks} 
We focus on the case where the aggregation rule is known. When the aggregation rule is unknown, we craft local models  based on an assumed one. Our empirical results in Section~\ref{unknownrule} show that our attacks have different levels of transferability between aggregation rules. 

\subsection{Optimization Problem}
\label{sec:opt}
\alan{Our idea is to manipulate the global model via carefully crafting the local models sent from the compromised worker devices to the master device in each iteration of  federated learning. 
We denote by $s_j$ the changing direction of the $j$th global model parameter in the current iteration when there are no attacks, where $s_j=1$ or $-1$.   $s_j=1$ (or $s_j=-1$) means that the $j$th global model parameter increases (or decreases) upon the previous iteration.  We consider the attacker's goal (we call it \emph{directed deviation goal}) is to deviate a global model parameter the most towards the inverse of the direction along which the global model parameter would change without attacks. 
Suppose in an iteration, $\mathbf{w}_i$ is the local model that the $i$th worker device intends to send to the master device when there are no attacks. Without loss of generality, we assume the first $c$ worker devices are compromised. 
Our directed deviation goal  is to craft local models $\mathbf{w}_1', \mathbf{w}_2', \cdots, \mathbf{w}_c'$ for the compromised worker devices via solving the following optimization problem in each iteration:
\begin{align}
\label{problem}
&\max_{\mathbf{w}_1', \cdots, \mathbf{w}_c'}  \mathbf{s}^T (\mathbf{w} - \mathbf{w}'),\nonumber \\
\text{subject to } & \mathbf{w}=\mathcal{A}(\mathbf{w}_1, \cdots, \mathbf{w}_c, \mathbf{w}_{c+1}, \cdots, \mathbf{w}_m), \nonumber \\
& \mathbf{w}'=\mathcal{A}(\mathbf{w}_1', \cdots, \mathbf{w}_c', \mathbf{w}_{c+1}, \cdots, \mathbf{w}_m),
\end{align}
where $\mathbf{s}$ is a column vector of the changing directions of all global model parameters, $\mathbf{w}$ is the before-attack global model, and $\mathbf{w}'$ is the after-attack  global model. \neil{Note that $\mathbf{s}$, $\mathbf{w}$, and $\mathbf{w}'$ all depend on the iteration number. Since our attacks manipulate the local models in each iteration, we omit the explicit dependency on the iteration number for simplicity.}

In our preliminary exploration of formulating poisoning attacks, we also considered a \emph{deviation goal}, which does not consider the global model parameters' changing directions. We empirically find that our attacks based on both the directed deviation goal and the deviation goal achieve high testing error rates for Krum. However, the directed deviation goal substantially outperforms the 
deviation goal for trimmed mean and median aggregation rules. Appendix~\ref{goals} shows our deviation goal and the empirical comparisons between deviation goal and directed deviation goal. }

\subsection{Attacking Krum} Recall that Krum selects one local model as the global model in each iteration. Suppose $\mathbf{w}$ is the selected local model in the current iteration when there are no attacks. Our goal is to craft the $c$ compromised local models such that the local model selected by Krum has the largest directed deviation from $\mathbf{w}$. Our  idea is to make Krum select a certain crafted local model (e.g., $\mathbf{w}_1'$ without loss of generality) via crafting the $c$ compromised local models. Therefore, we aim to solve the optimization problem in Equation~\ref{problem} with $\mathbf{w}'=\mathbf{w}_1'$ and the aggregation rule is Krum. 

\myparatight{Full knowledge} The key challenge of solving the optimization problem is that the constraint of the optimization problem is highly nonlinear and the search space of the local models $\mathbf{w}_1', \cdots, \mathbf{w}_c'$ is large. To address the challenge, we make two approximations.  Our approximations represent suboptimal solutions to the optimization problem, which means that the attacks based on the approximations may have suboptimal performance. However, as we will demonstrate in our experiments, our attacks already substantially increase the  error  rate of the learnt model. 

First, we restrict  $\mathbf{w}_{1}'$ as follows:  $\mathbf{w}_{1}' = \mathbf{w}_{Re} - \lambda \mathbf{s} $, where $\mathbf{w}_{Re}$ is the global model  received from the master device in the current iteration (i.e., the global model obtained in the previous iteration) and $\lambda > 0$. This approximation explicitly models the directed deviation between the crafted local model  $\mathbf{w}_{1}'$ and the received global model. 
We also explored the approximation $\mathbf{w}_{1}' = \mathbf{w} - \lambda \mathbf{s}$, which means that we explicitly model the directed deviation between the crafted local model and the local model selected by Krum before attack. However, we found that our attacks are less effective using this approximation.  

Second, to make $\mathbf{w}_1$ more likely to be selected by Krum, we craft the other $c-1$ compromised local models to be close to $\mathbf{w}_1'$. In particular, when the other $c-1$ compromised local models are close to $\mathbf{w}_1'$, $\mathbf{w}_1'$ only needs to have a small distance to $m-2c-1$ benign local models in order to be selected by Krum. In other words, the other $c-1$ compromised local models ``support'' the crafted local model $\mathbf{w}_1'$. In implementing our attack, we first assume the other $c-1$ compromised local models are the same as $\mathbf{w}_1'$, then we solve $\mathbf{w}_1'$, and finally we randomly sample $c-1$ vectors, whose distance to $\mathbf{w}_1'$ is at most $\epsilon$, as the other $c-1$ compromised local models. With our two approximations, we transform the optimization problem as follows:
\begin{align}
&\max_{\lambda} \lambda \nonumber \\
\label{problem2}
\text{subject to }&\mathbf{w}_1' = Krum(\mathbf{w}_1', \cdots, \mathbf{w}_c', \mathbf{w}_{(c+1)}, \cdots, \mathbf{w}_m), \nonumber \\
 			&\mathbf{w}_1'= \mathbf{w}_{Re} - \lambda \mathbf{s}, \nonumber \\
			& \mathbf{w}_i' = \mathbf{w}_1', \text{ for } i = 2, 3, \cdots, c.			
\end{align}
More precisely, the objective function in the above optimization problem should be $\mathbf{s}^T(\mathbf{w}-\mathbf{w}_{Re}) + \lambda \mathbf{s}^T\mathbf{s}$. However, $\mathbf{s}^T(\mathbf{w}-\mathbf{w}_{Re})$ is a constant and $\mathbf{s}^T\mathbf{s}=d$ where $d$ is the number of parameters in the global model. Therefore, we simplify the objective function to be just $\lambda$. 
After solving $\lambda$ in the optimization problem, we can obtain the crafted local model $\mathbf{w}_{1}'$. Then, we randomly sample $c-1$ vectors whose distance to $\mathbf{w}_1'$ is at most $\epsilon$ as the other $c-1$ compromised local models. We will explore the impact of $\epsilon$ on the effectiveness of our attacks in experiments. 

{\bf Solving $\lambda$.} Solving $\lambda$ in the optimization problem in Equation~\ref{problem2} is key to our attacks. First, we derive an upper bound of the solution $\lambda$ to the optimization problem. Formally, we have the following theorem. 

\begin{theorem}
\label{theoremLambda}
Suppose $\lambda$ is a solution to the optimization problem in Equation~\ref{problem2}. $\lambda$ is upper bounded as follows:
\begin{align}
\lambda \le & \sqrt{\frac{1}{(m-2c-1)d}  \cdot \min_{c+1\le i\le m}{\sum_{l\in {\tilde{\Gamma}_{\mathbf{w}_i}^{m-c-2}}}} D^2(\mathbf{w}_l,\mathbf{w}_i)} \nonumber\\
			&  +  \frac{1}{\sqrt{d}}\cdot \max_{c+1\le i\le m}{D(\mathbf{w}_i,\mathbf{w}_{Re})},
\end{align}
where $d$ is the number of parameters in the global model, $D(\mathbf{w}_l, \mathbf{w}_i)$ is the Euclidean distance between $\mathbf{w}_l$ and $\mathbf{w}_i $, $\tilde{\Gamma}_{w_i}^{m-c-2}$ is the set of $m-c-2$ benign local models that have the smallest Euclidean distance to  $\mathbf{w}_i$. 
\end{theorem}  

\begin{proof}
See Appendix~\ref{appendix:proof}.
\end{proof}

Given the upper bound, we use a binary search to solve $\lambda$. Specifically, we initialize $\lambda$ as the upper bound and check whether Krum selects $\mathbf{w}_{1}'$ as the global model; if not, then we half $\lambda$; we repeat this process until Krum selects $\mathbf{w}_{1}'$ or $\lambda$ is smaller than a certain threshold (this indicates that the optimization problem may not have a solution). In our experiments, we use $1\times10^{-5}$ as the threshold. 

\myparatight{Partial knowledge} In the partial knowledge scenario, the attacker does not know the local models on the benign worker devices, i.e., $\mathbf{w}_{(c+1)}, \cdots, \mathbf{w}_m$. As a result, the attacker does not know the changing directions $\mathbf{s}$ and cannot solve the optimization problem in Equation~\ref{problem2}. However, the attacker has access to the before-attack local models on the $c$ compromised worker devices. Therefore, we propose to craft compromised local models based on these before-attack local models. First, we compute the mean of the $c$ before-attack local models as  $\tilde{\mathbf{w}}=\frac{1}{c}\sum_{i=1}^c \mathbf{w}_i$. Second, we estimate the changing directions using the mean local model. Specifically, if the mean of the $j$th parameter is larger than the $j$th global model parameter received from the master device in the current iteration, then we estimate the changing direction for the $j$th parameter to be $1$, otherwise we estimate it to be $-1$. For simplicity, we denote by $\tilde{\mathbf{s}}$ the vector of estimated changing directions. 

Third, we treat the before-attack local models on the compromised worker devices as if they were local models on benign worker devices, and we aim to craft local model $\mathbf{w}_{1}'$ such that, among the crafted local model and the $c$ before-attack local models, Krum selects the crafted local model. 
 Formally, we have the following optimization problem:
\begin{align}
&\max_{\lambda} \lambda \nonumber \\
\label{problem3}
\text{subject to }&\mathbf{w}_1' = Krum(\mathbf{w}_1',  \mathbf{w}_{1}, \cdots, \mathbf{w}_c), \nonumber \\
 			&\mathbf{w}_1'= \mathbf{w}_{Re} - \lambda \tilde{\mathbf{s}}.
\end{align}

Similar to Theorem~\ref{theoremLambda}, we can also derive an upper bound of $\lambda$ for the optimization problem in Equation~\ref{problem3}. Moreover, similar to the full knowledge scenario, we use a binary search to solve $\lambda$. However, unlike the full knowledge scenario, if we cannot find a solution $\lambda$ until $\lambda$ is smaller than a threshold (i.e., $1\times10^{-5}$), then we add one more crafted local model $\mathbf{w}_2'$ such that among the crafted local models $\mathbf{w}_1'$, $\mathbf{w}_2'$, and the $c$ before-attack local models, Krum selects the crafted local model $\mathbf{w}_1'$. Specifically, we solve the optimization problem in Equation~\ref{problem3} with $\mathbf{w}_2'$ added into the Krum aggregation rule. Like the full knowledge scenario, we assume $\mathbf{w}_2'=\mathbf{w}_1'$. If we still cannot find a solution $\lambda$ until $\lambda$ is smaller than the threshold, we add another crafted local model. We repeat this process until finding a solution $\lambda$. We find that such iterative searching process makes our attack more effective for Krum in the partial knowledge scenario. 
After solving $\lambda$, we obtain the crafted local model $\mathbf{w}_1'$. Then, like the full knowledge scenario, we randomly sample  $c-1$ vectors whose distance to $\mathbf{w}_1'$ is at most $\epsilon$ as the other $c-1$ compromised local models.

\subsection{Attacking Trimmed Mean} 
Suppose ${w}_{ij}$ is the $j$th \emph{before-attack} local model parameter on the $i$th worker device and ${w}_{j}$ is the $j$th before-attack global model parameter in the current iteration. 
We discuss how we craft each local model parameter on the compromised worker devices.  
We denote by $w_{max,j}$ and $w_{min,j}$ the maximum and minimum of the $j$th local model parameters on the benign worker devices, i.e., $w_{max,j}$=$\max\{{w}_{(c+1)j}, {w}_{(c+2)j}, \cdots, {w}_{mj}\}$ and $w_{min,j}$=$\min\{{w}_{(c+1)j}, {w}_{(c+2)j}, \cdots, {w}_{mj}\}$. 

\myparatight{Full knowledge}
Theoretically, we can show that the following attack can maximize the directed deviations of the global model (i.e., an optimal solution to the optimization problem in Equation~\ref{problem}): if $s_j=-1$, 
then we use any $c$ numbers that are larger than $w_{max,j}$ as the $j$th local model parameters on the $c$ compromised worker devices, otherwise we use any $c$ numbers that are smaller than $w_{min,j}$ as the $j$th local model parameters on the $c$ compromised worker devices.

Intuitively, our attack crafts the compromised local models based on the maximum or  minimum benign local model parameters, depending on which one deviates the global model towards the inverse of the direction along which the global model would change without attacks. The sampled $c$ numbers should be close to $w_{max,j}$ or $w_{min,j}$ to avoid being outliers and being detected easily. Therefore, when implementing the  attack, if $s_j=-1$, then we randomly sample the $c$ numbers in the interval [$w_{max,j}, b\cdot w_{max,j}$] (when $w_{max,j} > 0$) or [$w_{max,j}, w_{max,j}/b$] (when $w_{max,j} \leq 0$), otherwise we randomly sample the $c$ numbers in the interval [$w_{min,j}/b, w_{min,j}$] (when $w_{min,j} > 0$) or [$b\cdot w_{min,j}, w_{min,j}$] (when $w_{min,j} \leq 0$). Our attack does not depend on $b$ once $b>1$. In our experiments, we set $b=2$.

\myparatight{Partial knowledge} An attacker faces two challenges in the partial knowledge scenario. First, the attacker does not know the changing direction variable $s_j$ because the attacker does not know the local models on the benign worker devices.  
Second, for the same reason, the attacker does not know the maximum $w_{max,j}$ and minimum $w_{min,j}$ of the benign local model parameters. 
Like Krum, to address the first challenge, we estimate the changing direction variables using the local models on the compromised worker devices. 

One naive strategy to address the second challenge is to use a very large number as $w_{max,j}$ or a very small number as $w_{min,j}$. However, if we craft the compromised local models based on $w_{max,j}$ or $w_{min,j}$ that are far away from their true values, the crafted local models may be outliers and the master device may detect the compromised local models easily. 
Therefore, we propose to estimate $w_{max,j}$ and $w_{min,j}$ using the before-attack local model parameters on the compromised worker devices. 
In particular, the attacker can compute the mean $\mu_j$ and standard deviation $\sigma_j$ of each $j$th parameter on the compromised worker devices. 

Based on the assumption that each $j$th  parameters of the benign worker devices are samples from a Gaussian distribution with mean $\mu_j$ and standard deviation $\sigma_j$, we can estimate that $w_{max,j}$ is smaller than $\mu_j + 3\sigma_j$ or $\mu_j + 4\sigma_j$ with large probabilities; and  $w_{min,j}$ is larger than $\mu_j - 4\sigma_j$ or $\mu_j - 3\sigma_j$ with large probabilities. Therefore, when $s_j$ is estimated to be $-1$, we sample $c$ numbers from the interval $[\mu_j + 3\sigma_j, \mu_j + 4\sigma_j]$ as the $j$th parameter of the $c$ compromised local models, which means that the crafted compromised local model parameters are larger than the maximum of the benign local model parameters with a high probability (e.g., 0.898 -- 0.998 when $m=100$ and $c=20$ under the Gaussian distribution assumption). When $s_j$ is estimated to be $1$, we sample $c$ numbers from the interval $[\mu_j - 4\sigma_j, \mu_j - 3\sigma_j]$ as the $j$th parameter of the $c$ compromised local models, which means that the crafted compromised local model parameters are smaller than the minimum of the benign local model parameters with a high probability.  The $j$th model parameters on the benign worker devices may not accurately follow a Gaussian distribution. However, our attacks are still effective empirically.

\subsection{Attacking Median} We use the same attacks for trimmed mean to attack the median aggregation rule. For instance, in the full knowledge scenario, 
 we randomly sample the $c$ numbers in the interval [$w_{max,j}, b\cdot w_{max,j}$]  or [$w_{max,j}, w_{max,j}/b$] if $s_j=-1$, otherwise we randomly sample the $c$ numbers in the interval [$w_{min,j}/b, w_{min,j}$]  or [$b\cdot w_{min,j}, w_{min,j}$]. 


\section{Evaluation}
\label{sec:exp}

\alan{We evaluate the effectiveness of our attacks using multiple datasets in different scenarios, e.g., the impact of different parameters and known vs. unknown aggregation rules. Moreover, we compare our attacks with existing attacks.}

\subsection{Experimental Setup}
\label{exp:setup}

\myparatight{Datasets} We consider four datasets: MNIST, Fashion-MNIST, CH-MNIST~\cite{kather2016multi}\footnote{We use a pre-processed version from \url{https://www.kaggle.com/kmader/colorectal-histology-mnist\#hmnist_64_64_L.csv}.} and Breast Cancer Wisconsin (Diagnostic) \cite{Dua:2019}. MNIST and Fashion-MNIST each includes  60,000 training examples and 10,000 testing examples, where each example is an 28$\times$28 grayscale image. 
Both datasets are 10-class classification problems. The CH-MNIST dataset consists of 5000 images of histology tiles from patients with colorectal cancer. The dataset is an  8-class classification problem. Each image has 64$\times$64 grayscale pixels. We randomly select 4000 images as the training examples and use the remaining 1000 as the testing examples. \alan{The Breast Cancer Wisconsin (Diagnostic) dataset is a binary classification problem to diagnose whether a person has breast cancer. The dataset contains 569 examples, each of which has 30 features describing the characteristics of a person's cell nuclei. We randomly select 455 (80\%) examples as the training examples, and use the remaining 114 examples as the testing examples.}

\myparatight{Machine learning classifiers} We consider the following classifiers.

{\bf Multi-class logistic regression (LR).} The considered aggregation rules have theoretical guarantees for the error rate of LR classifier. 

{\bf Deep neural networks (DNN).} For MNIST, Fashion-MNIST, and Breast Cancer Wisconsin (Diagnostic), we use a DNN with the architecture described in Table~\ref{DNN-architecture} in Appendix. We use ResNet20~\cite{he2016deep} for CH-MNIST.   Our DNN architecture does not necessarily achieve the smallest error rates for the considered datasets, as our goal is not to search for the best DNN architecture. Our goal is to show that our attacks can increase the testing error rates of the learnt DNN classifiers.

\myparatight{Compared attacks} We compare the following attacks.

{\bf Gaussian attack.} This attack randomly crafts the local models on the compromised worker devices. Specifically, for each $j$th model parameter, we estimate a Gaussian distribution using the before-attack local models on all worker devices. Then, for each compromised worker device, we sample a number from the Gaussian distribution and treat it as the $j$th parameter of the local model on the compromised worker device. We use this Gaussian attack to show that  crafting compromised local models randomly can not effectively attack the Byzantine-robust aggregation rules.

{\bf Label flipping attack.} This is a data poisoning attack that does not require  knowledge of the training data distribution. On each compromised worker device, this attack flips the label of each training instance. Specifically, we flip a label $l$ as $L-l-1$, where $L$ is the number of classes in the classification problem and $l=0,1,\cdots, L-1$. 

{\bf Back-gradient optimization based attack~\cite{munoz2017towards}.} This is the state-of-the-art untargeted data poisoning attack for multi-class classifiers. We note that this attack is not scalable and thus we compare our attacks with this attack on a subset of MNIST separately. The results are shown in Section~\ref{comparison-data}.  

{\bf Full knowledge attack or partial knowledge attack.} Our attack when the attacker knows the local models on all worker devices or the compromised ones.

\begin{table}[!t]\renewcommand{\arraystretch}{0.9}
\centering
\caption{\neil{Default setting for key parameters.}}
\addtolength{\tabcolsep}{-2pt}
\begin{tabular}{|c|c|c|} \hline 
{\small Parameter}  & {\small Description} & {\small Value}  \\ \hline
{\small  $m$}  & {\small Number of worker devices.} & {\small 100}  \\ \hline
{\small $c$}   & {\small Number of compromised worker devices.} & {\small 20}  \\ \hline
{\small $p$}   & {\small Degree of Non-IID.} & {\small 0.5}  \\ \hline
{\small $\epsilon$}   & {\small Distance parameter for Krum attacks.} & {\small 0.01}  \\ \hline
{\small $\beta$}   & {\small  Parameter of trimmed mean.} & {\small $c$}  \\ \hline
\end{tabular}
\label{defaultparameter}
\vspace{-2mm}
\end{table}

\myparatight{Parameter setting} We describe parameter setting for the federated learning algorithms and our attacks. Table~\ref{defaultparameter} summarizes the default setting for key parameters. \alan{We use MXNet~\cite{chen2015mxnet} to implement federated learning and attacks. We repeat each experiment for 50 trials and report the average results. We observed that the variances are very small, so we omit them for simplicity.}
\label{parametersetting}

{\bf Federated learning algorithms.} By default, we assume $m=100$ worker devices; each worker device applies one round of stochastic gradient descent to update its local model; and the master device aggregates local models from all worker devices. 
One unique characteristic of federated learning is that the local training datasets on different devices may not be \emph{independently and identically distributed} (i.e., \emph{non-IID})~\cite{McMahan17}. We simulate federated learning with different non-IID training data distributions. Suppose we have $L$ classes in the classification problem, e.g., $L=10$ for the MNIST and Fashion-MNIST datasets, and $L=8$ for the CH-MNIST dataset. We evenly split the worker devices into $L$ groups. \alan{We model non-IID federated learning by assigning a training instance with label $l$ to the $l$th group with probability $p$, where $p > 0$.}  
A higher $p$ indicates a higher degree of non-IID. For convenience, we call the probability $p$ \emph{degree of non-IID}.  Unless otherwise mentioned, we set $p=0.5$. 

We set 500 iterations for the LR classifier on  MNIST;  we set 2,000 iterations for the DNN classifiers on all four datasets; and we set the batch size to be 32 in stochastic gradient descent, except that we set the batch size to be 64 for Fashion-MNIST as such setting leads to a more accurate model.  
The trimmed mean aggregation rule prunes the largest and smallest $\beta$ parameters, where $ c \leq \beta < \frac{m}{2}$. Pruning more parameters leads to larger testing error rates without attacks. By default, we consider $\beta=c$ as the authors of trimmed mean did~\cite{Yin18}.

{\bf Our attacks.} Unless otherwise mentioned, we consider 20 worker devices are compromised. Our attacks to Krum have a parameter $\epsilon$, which is related to the distance between the crafted compromised local models. We set $\epsilon=0.01$ (we will study the impact of $\epsilon$ on our attack). We do not set $\epsilon=0$ because $\epsilon=0$ makes the $c$ compromised local models exactly the same, making the compromised local models easily detected by the master device. 
Our attacks to trimmed mean and median have a parameter $b$ in the full knowledge scenario, where $b>1$. Our attacks do not depend on $b$ once  $b>1$. 
We set $b=2$. \neil{Unless otherwise mentioned, we assume that attacker manipulates the local models on the compromised worker devices in each iteration.}

\begin{table}[!t]\renewcommand{\arraystretch}{1}
\centering
\caption{\neil{Testing error rates of various attacks.}}
\vspace{-1mm}
\centering
\addtolength{\tabcolsep}{-4pt}
\subfloat[LR classifier, MNIST]{
\begin{tabular}{|c|c|c|c|c|c|} \hline 
{\small } & {\small NoAttack} & {\small Gaussian} & {\small LabelFlip} &{\small Partial} & {\small Full}  \\ \hline
{\small Krum}  & {\small 0.14} & {\small 0.13} & {\small 0.13} &{\small 0.72} & {\small 0.80}  \\ \hline
{\small Trimmed mean}  & {\small 0.12} & {\small 0.11} & {\small 0.13} &{\small 0.23} & {\small 0.52}  \\  \hline
{\small Median}  & {\small 0.13} & {\small 0.13} & {\small 0.15} &{\small 0.19} & {\small 0.29}  \\ \hline 
\end{tabular}
\label{attackeffective-MNIST-LR}
}

\subfloat[DNN classifier, MNIST]{
\begin{tabular}{|c|c|c|c|c|c|} \hline 
{\small } & {\small NoAttack} & {\small Gaussian} & {\small LabelFlip} &{\small Partial} & {\small Full}  \\ \hline
{\small Krum}  & {\small 0.11} & {\small 0.10} & {\small 0.10} &{\small 0.75} & {\small 0.77}  \\ \hline
{\small Trimmed mean}  & {\small 0.06} & {\small 0.07} & {\small 0.07} &{\small 0.14} & {\small 0.23}  \\  \hline
{\small Median}  & {\small 0.06} & {\small 0.06} & {\small 0.16} &{\small 0.28} & {\small 0.32}  \\ \hline
\end{tabular}
\label{attackeffective-MNIST-DNN}
}

\subfloat[DNN classifier, Fashion-MNIST]{
	\begin{tabular}{|c|c|c|c|c|c|} \hline 
	{\small } & {\small NoAttack} & {\small Gaussian} & {\small LabelFlip} &{\small Partial} & {\small Full}  \\ \hline
	{\small Krum}  & {\small 0.16} & {\small 0.16} & {\small 0.16} &{\small 0.90} & {\small 0.91}  \\ \hline
	{\small Trimmed mean}  & {\small 0.10} & {\small 0.10} & {\small 0.12} &{\small 0.26} & {\small 0.28}  \\ \hline
	{\small Median}  & {\small 0.09} & {\small 0.12} & {\small 0.12} &{\small 0.21} & {\small 0.29}  \\ \hline
\end{tabular}
\label{attackeffective-fashion}
}

\subfloat[DNN classifier, CH-MNIST]{
\begin{tabular}{|c|c|c|c|c|c|} \hline 
{\small } & {\small NoAttack} & {\small Gaussian} & {\small LabelFlip} &{\small Partial} & {\small Full}  \\ \hline
{\small Krum}  & {\small 0.29} & {\small 0.30} & {\small 0.43} &{\small 0.73} & {\small 0.81}  \\ \hline
{\small Trimmed mean}  & {\small 0.17} & {\small 0.25} & {\small 0.37} &{\small 0.69} & {\small 0.69}  \\ \hline
{\small Median}  & {\small 0.17} & {\small 0.20} & {\small 0.17} &{\small 0.57} & {\small 0.63}  \\ \hline
\end{tabular}
\label{attackeffective-CHMNIST}}

\subfloat[DNN classifier, Breast Cancer Wisconsin (Diagnostic)]{
\begin{tabular}{|c|c|c|c|c|c|} \hline 
{\small } & {\small NoAttack} & {\small Gaussian} & {\small LabelFlip} &{\small Partial} & {\small Full}  \\ \hline
{\small Krum}  & {\small 0.03} & {\small 0.04} & {\small 0.14} &{\small 0.17} & {\small 0.17}  \\ \hline
{\small Trimmed mean}  & {\small 0.02} & {\small 0.03} & {\small 0.05} &{\small 0.14} & {\small 0.15}  \\ \hline
{\small Median}  & {\small 0.03} & {\small 0.03} & {\small 0.04} &{\small 0.17} & {\small 0.18}  \\ \hline
\end{tabular}
\label{attackeffective-BreastCancer}}
\vspace{-8mm}
\label{overallresults}
\end{table}

\subsection{Results for Known Aggregation Rule}
\label{sec:knownrule}

\begin{figure*}[!t]
\centering
\subfloat[Krum]{\includegraphics[width=0.33 \textwidth]{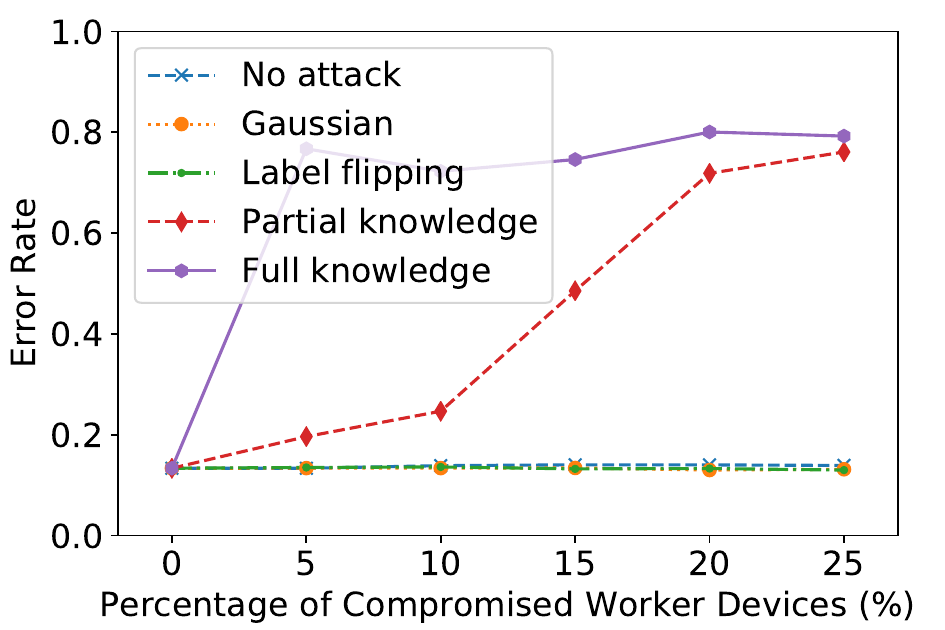}\label{label:MNIST-LR-Krum-federated}}
\subfloat[Trimmed mean]{\includegraphics[width=0.33 \textwidth]{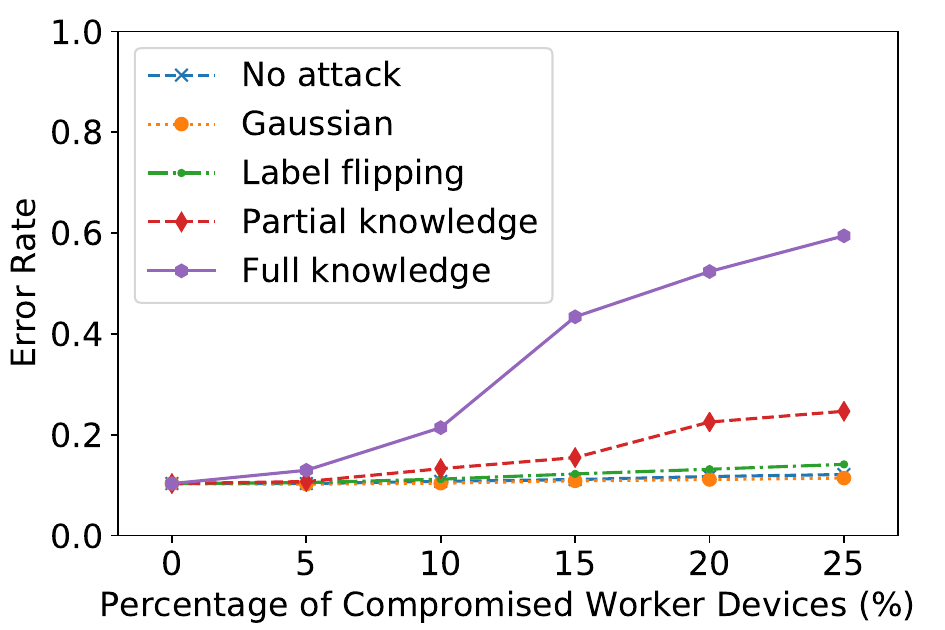}}
\subfloat[Median]{\includegraphics[width=0.33 \textwidth]{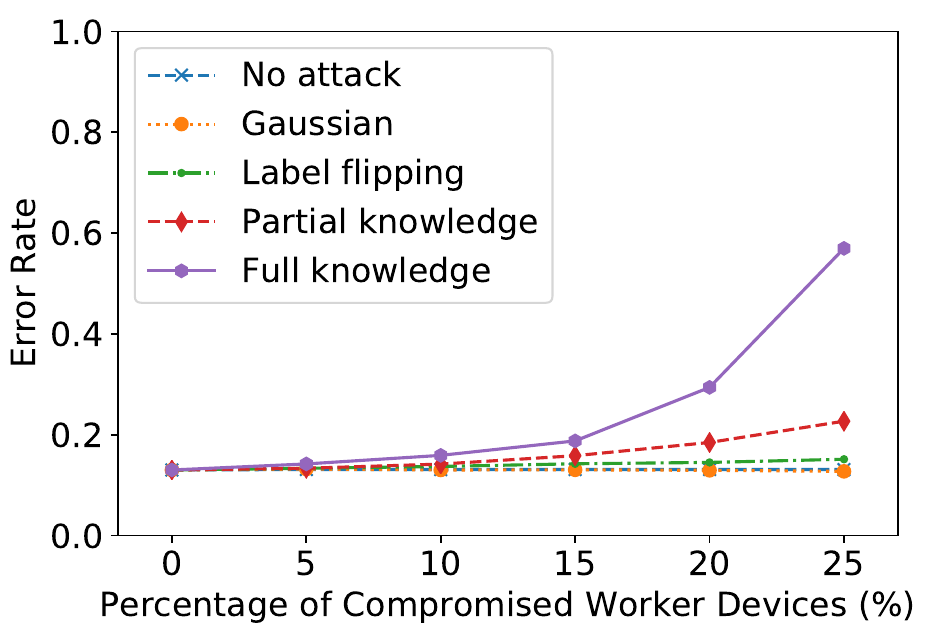}}
\vspace{-2mm}
\subfloat[Krum]{\includegraphics[width=0.33 \textwidth]{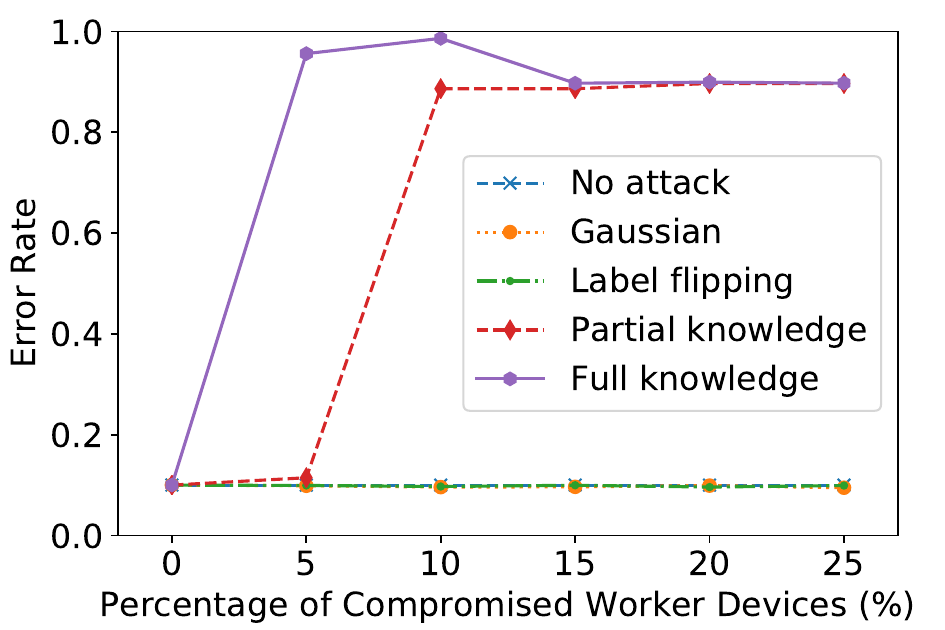}\label{label:MNIST-CNN-Krum-federated}}
\subfloat[Trimmed mean]{\includegraphics[width=0.33 \textwidth]{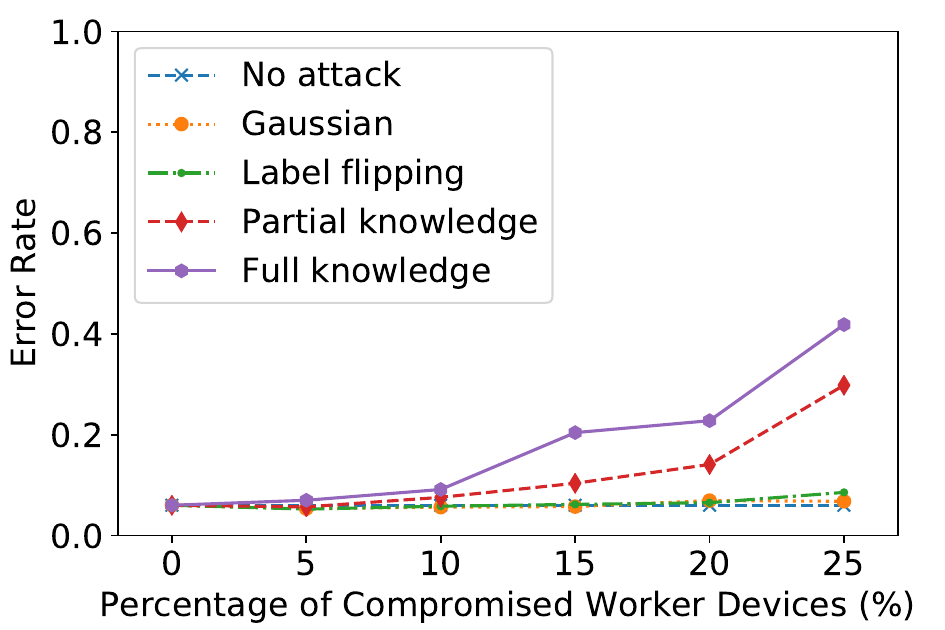}}
\subfloat[Median]{\includegraphics[width=0.33 \textwidth]{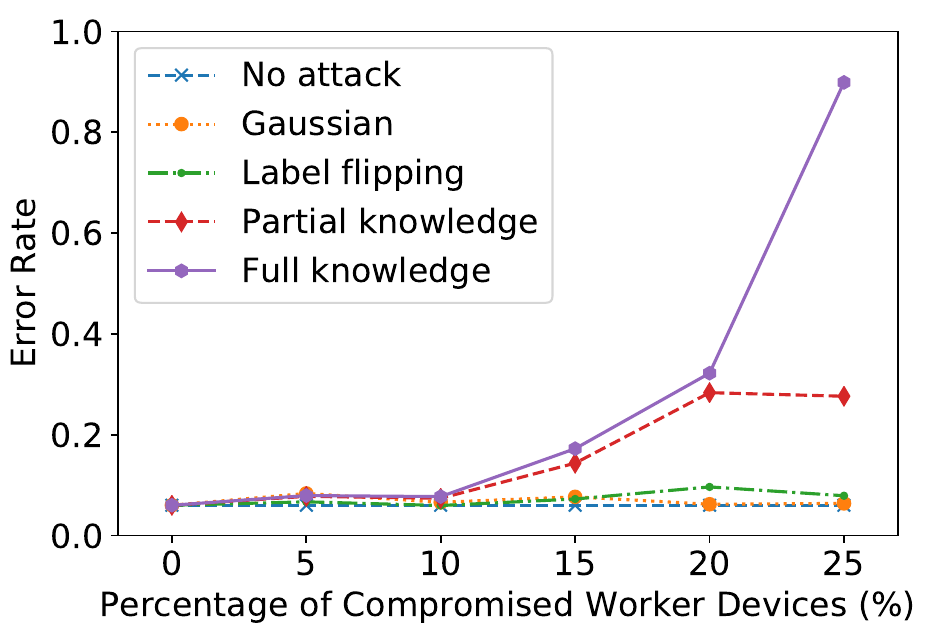}}
\caption{Testing error rates for different attacks as we have more compromised worker devices on MNIST. (a)-(c): LR classifier and (d)-(f): DNN classifier.}
\vspace{-4mm}
\label{MNIST-worker}
\end{figure*}

\noindent
\neil{{\bf Our attacks are effective:} Table~\ref{overallresults} shows the testing error rates of the compared attacks on the four datasets.  
First, these results show that our attacks are effective and substantially outperform existing attacks, i.e., our attacks result in higher error rates. For instance, when dataset is  MNIST,   classifier is LR, and aggregation rule is Krum, our partial knowledge attack increases the error rate from  0.14 to 0.72 (around 400\% relative increase). 
 Gaussian attacks only increase the error rates in several cases, e.g., median aggregation rule for Fashion-MNIST, and  trimmed mean and median for CH-MNIST.  Label flipping attacks can increase the error rates for DNN classifiers in some cases but have limited success for LR classifiers. } 

Second, Krum is less robust to our attacks than trimmed mean and median, except on Breast Cancer Wisconsin (Diagnostic) where Krum is comparable to median.   A possible reason why trimmed mean and median outperform Krum is that Krum picks one local model as the global model, while trimmed mean and median aggregate multiple local models to update the global model (the median selects one local model parameter for each model parameter, but the selected parameters may be from different local models). Trimmed mean is more robust to our attacks in some cases while median is more robust in other cases. \alan{Third, we observe that the error rates may depend on the data dimension. For instance, MNIST and Fashion-MNIST have 784 dimensions, CH-MNIST has 4096 dimensions, and Breast Cancer Wisconsin (Diagnostic) has 30 dimensions. For the DNN classifiers, the error rates are higher on CH-MNIST  than on other datasets in most cases, while the error rates are lower on  Breast Cancer Wisconsin (Diagnostic)  than on other datasets in most cases.}  

\alan{We note that federated learning may have higher error rate than centralized learning, even if robustness feature is not considered (i.e., mean aggregation rule is used). For instance, {the DNN classifiers respectively achieve testing error rates 0.01, 0.08, 0.07, and 0.01 in centralized learning on the four datasets,} while they respectively achieve testing error rates 0.04, 0.09, 0.09, and 0.01 in federated learning with the mean aggregation rule on the four datasets. 
However,  in the scenarios where users' training data can only be stored on their edge/mobile devices, e.g., for privacy purposes, centralized learning is not applicable and federated learning may be the only option even though its error rate is higher. Compared to the mean aggregation rule, Byzantine-robust aggregation rule increases the error rate without attacks. However, if Byzantine-robust aggregation rule is not used, a single malicious device can make the learnt global model totally useless~\cite{Blanchard17,Yin18}. To summarize, in the scenarios where users' training data can only be stored on their edge/mobile devices and there may exist attacks, Byzantine-robust federated learning may be the best option, even if its error rate is higher. }

\begin{figure*}[!t]
\centering
\subfloat[Krum]{\includegraphics[width=0.33 \textwidth]{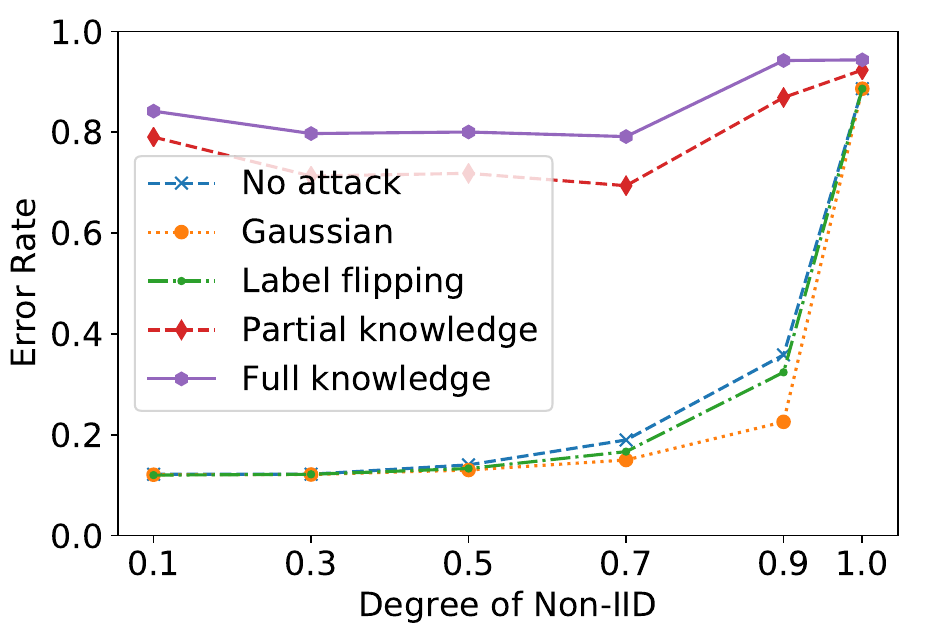}}
\subfloat[Trimmed mean]{\includegraphics[width=0.33 \textwidth]{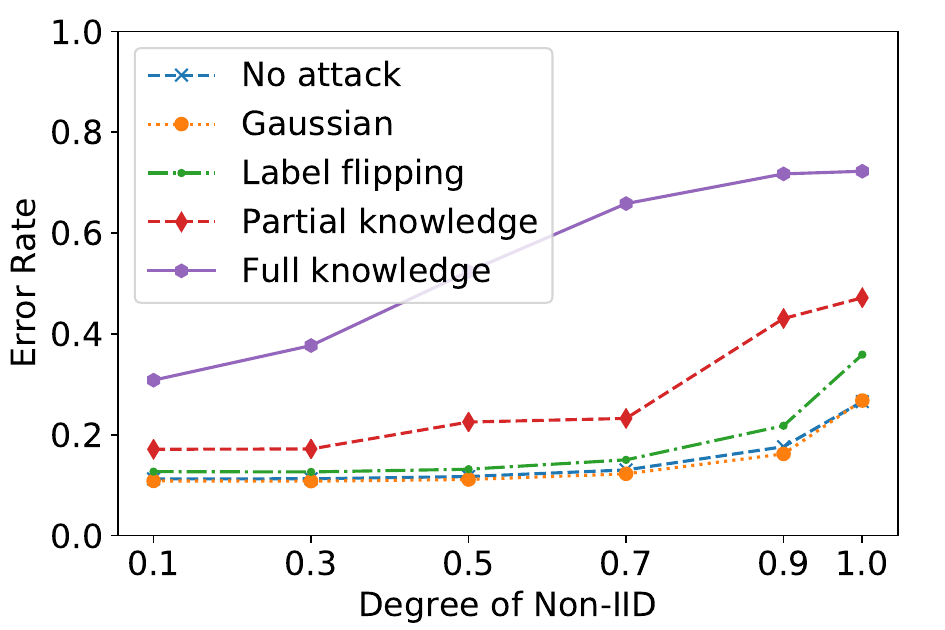}}
\subfloat[Median]{\includegraphics[width=0.33 \textwidth]{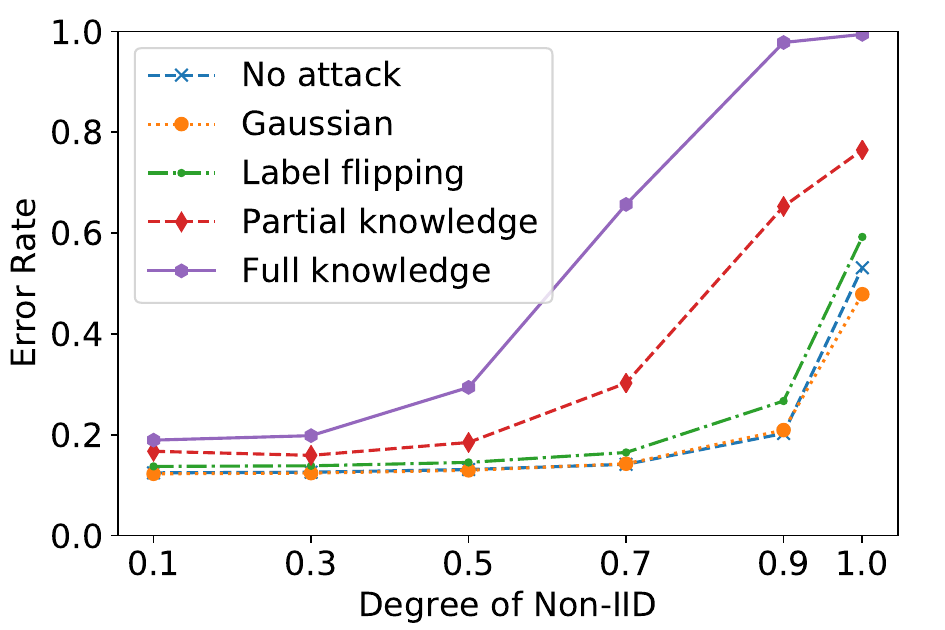}}

\subfloat[Krum]{\includegraphics[width=0.33 \textwidth]{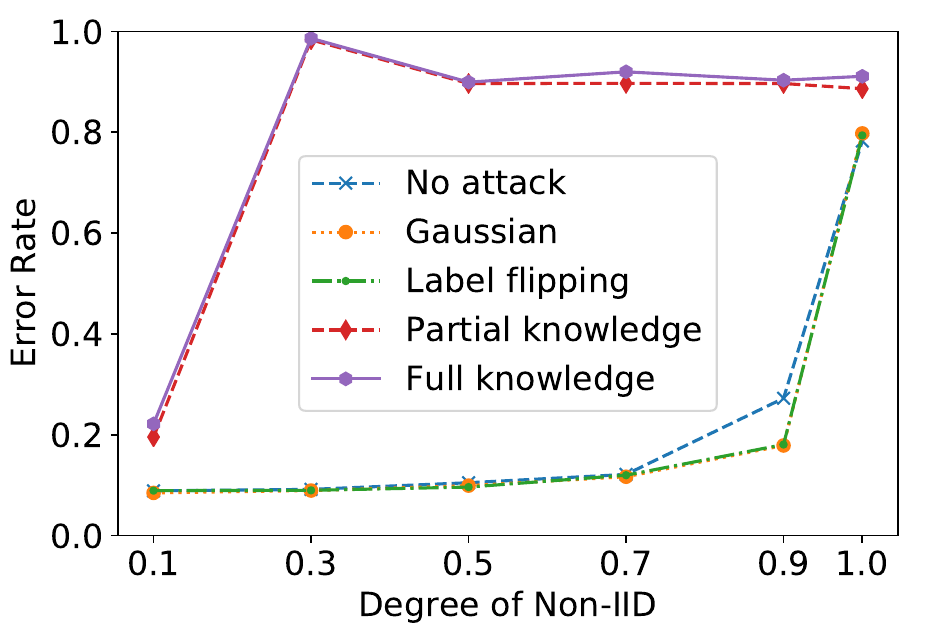}}
\subfloat[Trimmed mean]{\includegraphics[width=0.33 \textwidth]{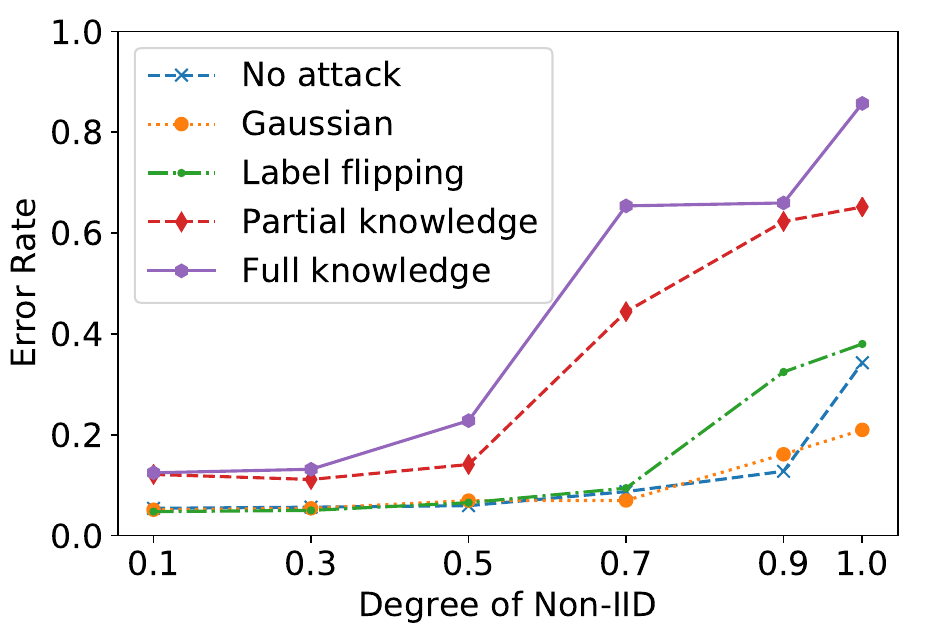}}
\subfloat[Median]{\includegraphics[width=0.33 \textwidth]{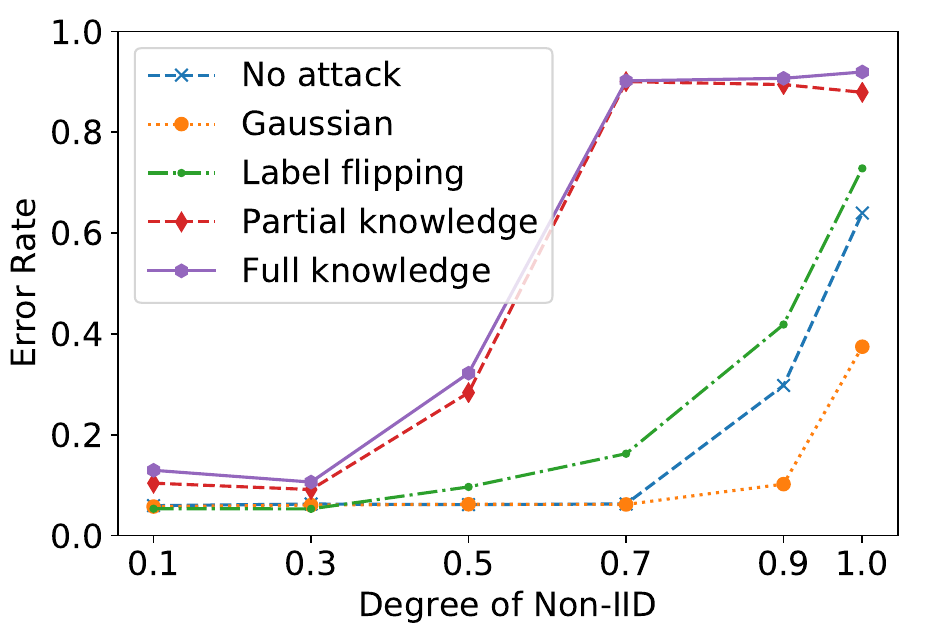}}
\vspace{-2mm}
\caption{Testing error rates for different attacks as we increase  the degree of non-IID on MNIST.  (a)-(c): LR classifier and (d)-(f): DNN classifier.}
\vspace{-3mm}
\label{MNIST-noniid}
\end{figure*}

\myparatight{Impact of the percentage of compromised worker devices} Figure~\ref{MNIST-worker} shows the error rates of different attacks as the percentage of compromised worker devices increases on MNIST.  
 Our attacks increase the error rates significantly as we compromise more worker devices; label flipping only slightly increases the error rates; and Gaussian attacks have no notable impact on the error rates. Two exceptions are that Krum's error rates decrease when the percentage of compromised worker devices increases from 5\% to 10\% in Figure~\ref{label:MNIST-LR-Krum-federated} and from 10\% to 15\% in Figure~\ref{label:MNIST-CNN-Krum-federated}. We suspect the reason is that Krum selects one local model as a global model in each iteration. We have similar observations on the other datasets. Therefore, we omit the corresponding results for simplicity.

\myparatight{Impact of the degree of non-IID in federated learning} Figure~\ref{MNIST-noniid} shows the error rates for the compared attacks for different degrees of non-IID on MNIST.  
 Error rates of all attacks including no attacks increase as we increase the degree of non-IID, except that the error rates of our attacks to Krum  fluctuate as the degree of non-IID increases. A possible reason is that as the local training datasets on different worker devices are more non-IID, the local models are more diverse, leaving more room for attacks. For instance, an extreme example is that if the local models on the benign worker devices are the same, it would be harder to attack the aggregation rules, because their aggregated model would be more likely to depend on the benign local models.

\begin{figure*}[!t]
\centering
\vspace{-2mm}
\subfloat[]{\includegraphics[width=0.33 \textwidth]{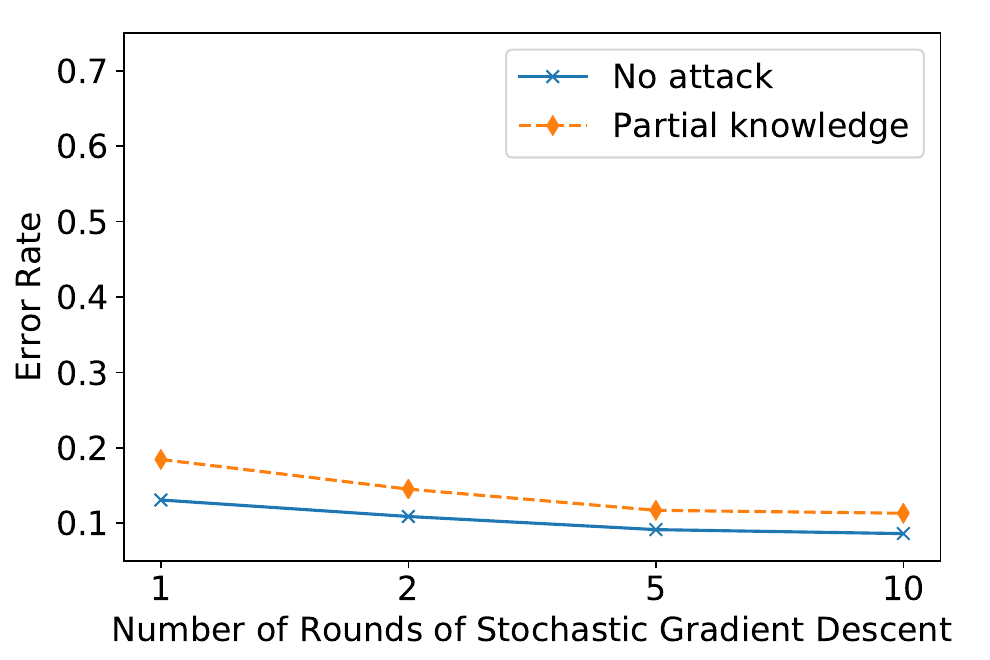}\label{impact:round}}
\subfloat[]{\includegraphics[width=0.33 \textwidth]{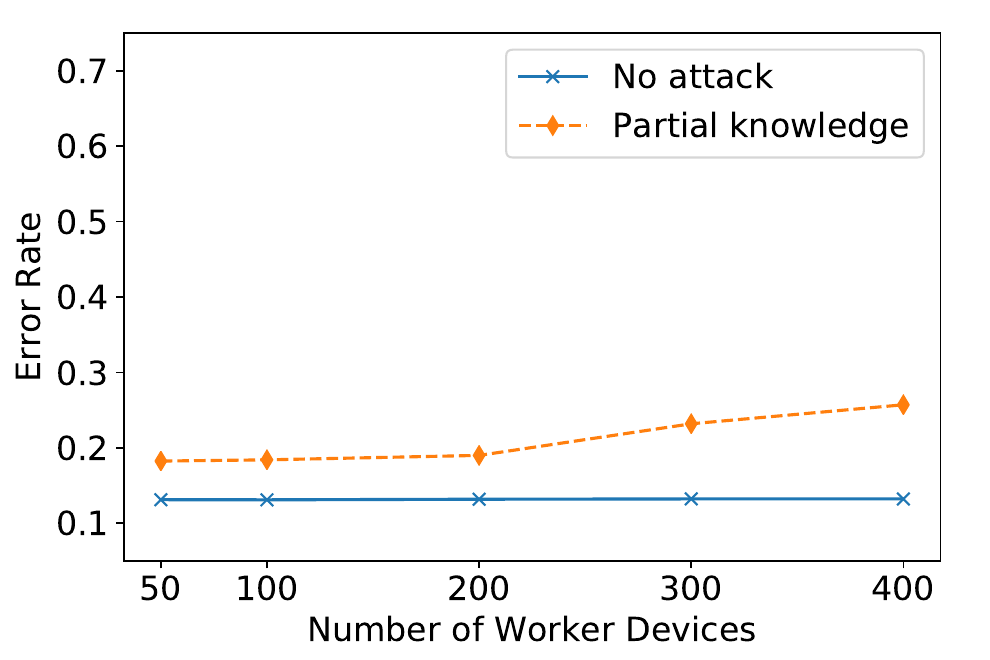}\label{impact:worker}}
\subfloat[]{\includegraphics[width=0.33 \textwidth]{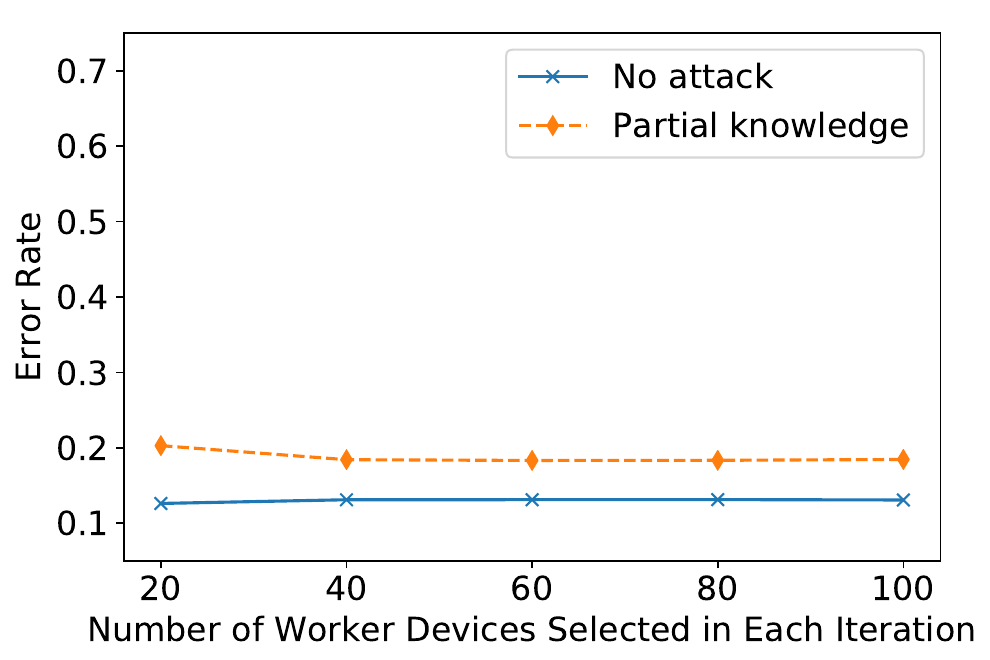}\label{impact:subset}}
\vspace{-2mm}
\caption{\alan{(a) Impact of the number of rounds of stochastic gradient descent worker devices use to update their local models in each iteration on our attacks. (b) Impact of the number of worker devices on our attacks. (c) Impact of the number of worker devices selected in each iteration on our attacks. MNIST, LR classifier, and median are used.}}
\vspace{-3mm}
\label{MNIST-LR-parasetting}
\end{figure*}

\begin{figure*}[!t]
\centering
\subfloat[]{\includegraphics[width=0.33 \textwidth]{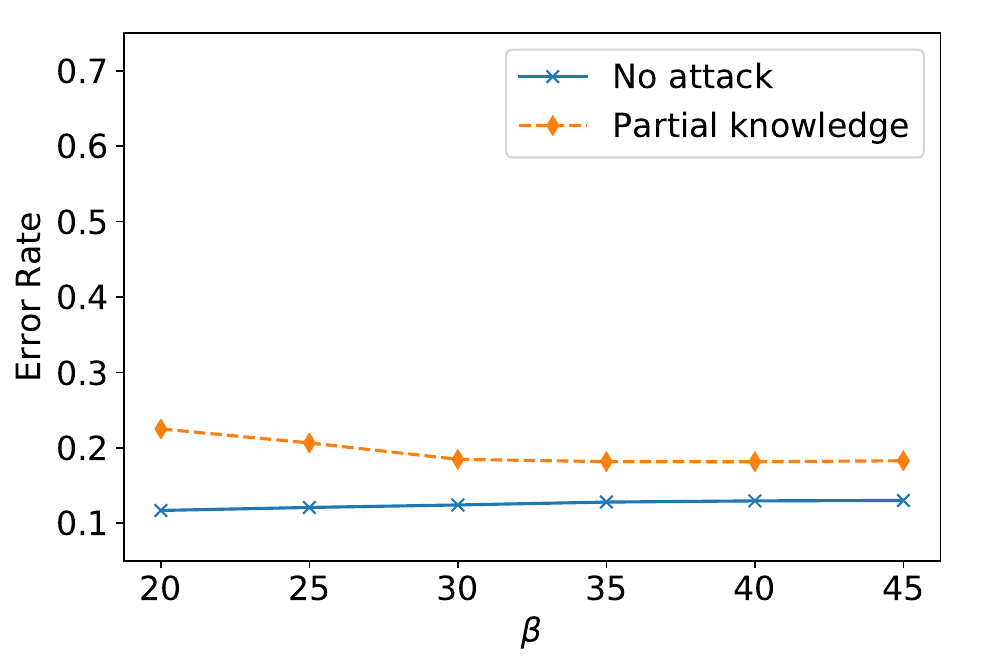}\label{MNIST-LR-Trim-beta}}
\subfloat[]{\includegraphics[width=0.33 \textwidth]{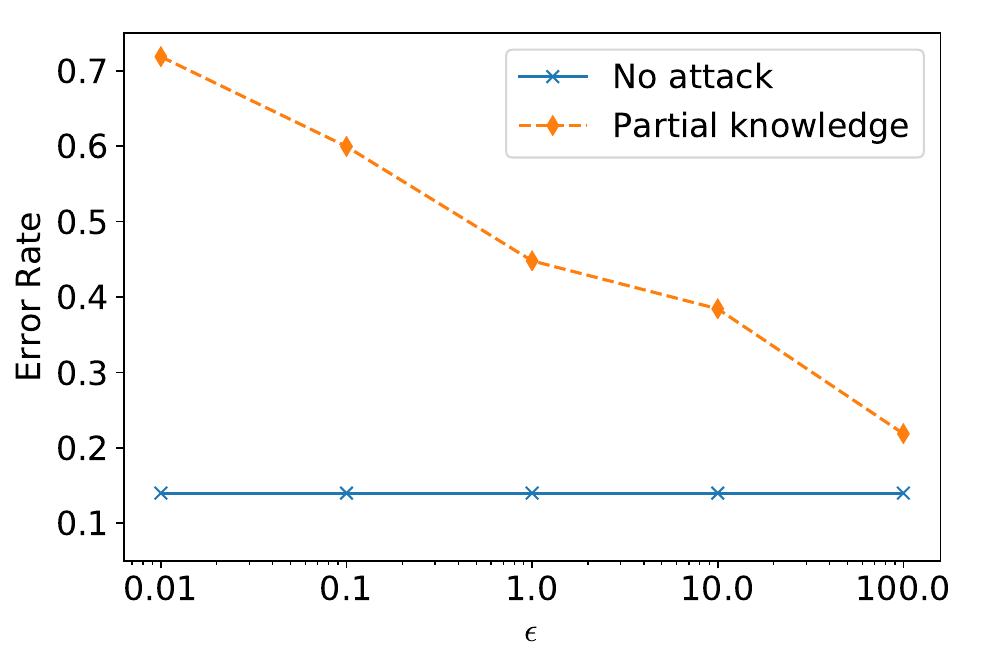}\label{MNIST-LR-Krum-epsilon}}
\subfloat[]{\includegraphics[width=0.33 \textwidth]{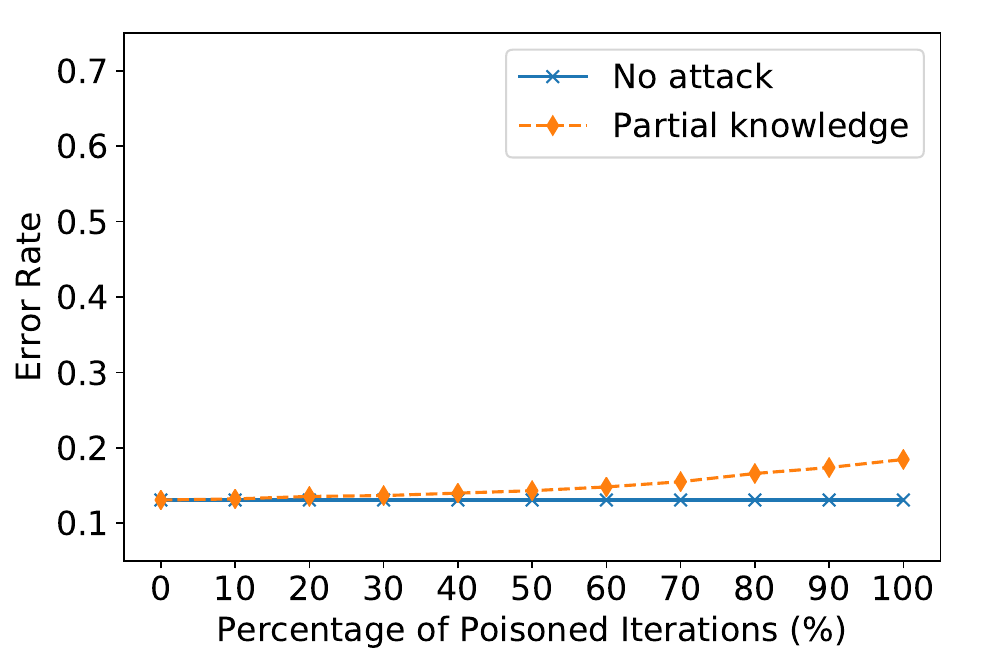}\label{MNIST-LR-Median-randround}}
\vspace{-2mm}
\caption{\alan{(a) Testing error rates of the trimmed mean aggregation rule when using different $\beta$. (b) Testing error rates of the Krum aggregation rule when our attack uses different $\epsilon$. (c) Testing error rates of the median aggregation rule when our attacks poison a certain fraction of randomly selected iterations of federated learning. MNIST and LR classifier are used.}}
\vspace{-3mm}
\label{MNIST-LR-parasetting1}
\end{figure*}

\myparatight{Impact of different parameter settings of federated learning algorithms} We study the impact of various parameters in federated learning including the number of rounds of stochastic gradient descent each worker device performs, number of worker devices,  number of worker devices selected to update the global model in each iteration, and $\beta$ in trimmed mean. \alan{In these experiments, we use MNIST and the LR classifier for simplicity. Unless otherwise mentioned, we consider  median, as median is more robust than Krum and does not require configuring extra parameters (trimmed mean requires configuring  $\beta$). Moreover, for simplicity, we consider partial knowledge attacks as they are more practical.}   

Worker devices can perform multiple rounds of stochastic gradient descent to update their local models. Figure~\ref{impact:round} shows the impact of the number of rounds on the testing error rates of our attack. The testing error rates decrease as we use more rounds of stochastic gradient descent for both no attack and our partial knowledge attack. This is because more rounds of stochastic gradient descent lead to more accurate local models, and  the local models on different worker devices are less diverse, leaving a smaller attack space. However, our attack still increases the error rates substantially even if we use more rounds. For instance, our attack still increases the error rate by more than 30\% when using 10 rounds of stochastic gradient descent. We note that a large number of rounds result in large computational cost for worker devices, which may be unacceptable for resource-constrained devices such as mobile phones and IoT devices.
 
Figure~\ref{impact:worker} shows the testing error rates of our attack as the number of worker devices increases, where 20\% of worker devices are compromised. Our attack is more effective (i.e., testing error rate is larger) as the federated learning system involves more worker devices. \alan{We found a possible reason is that our partial knowledge attacks can more accurately estimate the changing directions with more worker devices. For instance, for trimmed mean of the DNN classifier on MNIST, our partial knowledge attacks can correctly estimate the changing directions of 72\% of the global model parameters on average when there are 50 worker devices, and this fraction increases to 76\% when there are 100 worker devices.}  

In federated learning~\cite{McMahan17}, the master device could randomly sample some worker devices and send the global model to them; the sampled worker devices update their local models and send the updated local models to the master device; and the master device updates the global model using the local models from the sampled worker devices. Figure~\ref{impact:subset} shows the impact of the number of worker devices selected in each iteration on the testing error rates of our attack, where the total number of worker devices is 100. \alan{Since the master device randomly selects a subset of worker devices  in each iteration, a smaller number of compromised worker devices are selected in some iterations, while a larger number of compromised worker devices are selected in other iterations. On average, among the selected worker devices, $\frac{c}{m}$ of them are compromised ones, where $c$ is the total number of compromised worker devices and $m$ is the total number of worker devices. Our Figure~\ref{MNIST-worker} shows that our attacks become effective when $\frac{c}{m}$ is larger than 10\%-15\%.  Note that an attacker can inject a large number of fake devices to a federated learning system, so $\frac{c}{m}$ can be large.}

The trimmed mean aggregation rule has a parameter $\beta$, which should be at least the number of compromised worker devices. Figure~\ref{MNIST-LR-Trim-beta} shows the testing error rates of no attack and our partial knowledge attack as $\beta$ increases. Roughly speaking, our attack is less effective (i.e., testing error rates are smaller) as more local model parameters are trimmed. This is because our crafted local model parameters on the compromised worker devices are more likely to be trimmed when the master device trims more local model parameters. However, the testing error of no attack also \alan{slightly} increases as $\beta$ increases. 
The reason is that more benign local model parameters are trimmed and the mean of the remaining local model parameters becomes less accurate. The master device may be motivated to use a smaller $\beta$ to guarantee performance when there are no attacks.

\myparatight{Impact of the parameter $\epsilon$ in our attacks to Krum} Figure~\ref{MNIST-LR-Krum-epsilon} shows the error rates of the Krum aggregation rule when our attacks use different $\epsilon$, where MNIST dataset and LR classifier are considered. We observe that our attacks can effectively increase the error rates using a wide range of $\epsilon$. Moreover, our attacks achieve larger error rates when $\epsilon$ is smaller. This is because when $\epsilon$ is smaller, the distances between the compromised local models are smaller, which makes it more likely for Krum to select the local model crafted by our attack as the global model. 

\myparatight{Impact of the number of poisoned iterations} Figure~\ref{MNIST-LR-Median-randround} shows the error rates of the median aggregation rule when our attacks poison the local models on the compromised worker devices in a certain fraction of randomly selected iterations of federated learning. Unsurprisingly, the error rate increases when poisoning more iterations.

\begin{table}[!t]\renewcommand{\arraystretch}{1}
\centering
\caption{\neil{Testing error rates of attacks on the DNN classifier for MNIST when the master device chooses the global model with the lowest testing error rate.}}
\addtolength{\tabcolsep}{-4pt}
{
\begin{tabular}{|c|c|c|c|c|c|} \hline 
{\small } & {\small NoAttack} & {\small Gaussian} & {\small LabelFlip} &{\small Partial} & {\small Full}  \\ \hline
{\small Krum}  & {\small 0.10} & {\small 0.10} & {\small 0.09} &{\small 0.69} & {\small 0.70}  \\ \hline
{\small Trimmed mean}  & {\small 0.06} & {\small 0.06} & {\small 0.07} &{\small 0.12} & {\small 0.18}  \\  \hline
{\small Median}  & {\small 0.06} & {\small 0.06} & {\small 0.06} &{\small 0.11} & {\small 0.32}  \\ \hline
\end{tabular}
}
\label{adaptivetraining}
\vspace{-3mm}
\end{table}

\neil{\myparatight{Alternative training strategy} Each iteration results in a global model. Instead of selecting the last global model as the final model, an alternative training strategy is to select the global model that has the lowest \alan{testing} error rate.\footnote{\alan{We give advantages to the alternative training strategy since we use testing error rate to select the global model.}} Table~\ref{adaptivetraining} shows the testing error rates of various attacks on the DNN classifier for MNIST, when such alternative training strategy is adopted. \alan{In these experiments, our attacks attack each iteration of federated learning, and the column ``NoAttack'' corresponds to the scenarios where no iterations are attacked.} Compared to Table~\ref{attackeffective-MNIST-DNN}, this alternative training strategy is \alan{slightly} more secure against our attacks. However,  our attacks are still effective. For instance, for the Krum, trimmed mean, and median aggregation rules, our partial knowledge attacks still increase the testing error rates by 590\%, 100\%, and 83\%, respectively. } \alan{Another training strategy is to roll back to a few iterations ago if the master device detects an unusual increase of training error rate. However, such training strategy is not applicable because the training error rates of the global models  still decrease until convergence  when we perform our attacks in each iteration. In other words, there are no unusual increases of training error rates.}

\vspace{-4mm}
\subsection{Results for Unknown Aggregation Rule}
\label{unknownrule}

We craft local models based on one aggregation rule and show the attack effectiveness for other aggregation rules. 
Table~\ref{transfer_result} shows the transferability between aggregation rules, where MNIST and LR classifier are considered.  
   We observe different levels of transferability between aggregation rules. Specifically, Krum based attack can well transfer to trimmed mean and median, e.g., Krum based attack increases the error rate from 0.12 to 0.15 (25\% relative increase) for trimmed mean, and from 0.13 to 0.18 (38\% relative increase) for median. Trimmed mean based attack does not transfer to Krum but  transfers to median well. For instance, trimmed mean based attack increases the error rates from 0.13 to 0.20 (54\% relative increase) for median.

\begin{table}[!t]\renewcommand{\arraystretch}{1}
\centering
\caption{Transferability between aggregation rules. ``Krum attack'' and ``Trimmed mean attack'' mean that we craft the compromised local models based on the Krum and trimmed mean aggregation rules, respectively. Partial knowledge attacks are considered. The numbers are testing error rates.}
\centering
{
\begin{tabular}{|c|c|c|c|} \hline 
{\small } & {\small Krum} & {\small Trimmed mean} & {\small Median}  \\ \hline
{\small No attack} & {\small 0.14} & {\small 0.12} & {\small 0.13}\\ \hline
{\small Krum attack} & {\small 0.70} & {\small 0.15} & {\small 0.18}\\ \hline
{\small Trimmed mean attack} & {\small 0.14} & {\small 0.25} & {\small 0.20}\\ \hline
\end{tabular}
}
\vspace{-4mm}
\label{transfer_result}
\end{table}

\subsection{Comparing with Back-gradient Optimization based Attack}
\label{comparison-data}
Back-gradient optimization based attack (BGA)~\cite{munoz2017towards}  
 is state-of-the-art untargeted data poisoning attack for multi-class classifiers such as multi-class LR and DNN. \alan{BGA formulates a bilevel optimization problem, where the inner optimization is to minimize the training loss on the poisoned training data and the outer optimization is to find poisoning examples that maximize the minimal training loss in the inner optimization. BGA iteratively finds the poisoned examples by alternately solving the inner minimization and outer maximization problems.} 
We implemented BGA and  verified that our implementation can reproduce the results reported by the authors. 
 However, BGA is not scalable to the entire MNIST dataset. Therefore, we uniformly sample 6,000 training examples in MNIST, and we learn a 10-class LR classifier. Moreover, we assume 100 worker devices, randomly distribute the 6,000 examples to them, and assume 20 worker devices are compromised. 

\myparatight{Generating poisoned data} We assume an attacker has \emph{full knowledge} about the training datasets on all worker devices. Therefore, the attacker can use BGA to generate poisoned data based on the 6,000 examples. In particular, we run the attack for 10 days  on a GTX 1080Ti GPU, which generates 240 ($240/6000=4\%$) poisoned examples. We verified that these poisoned data can effectively increase the testing error rate if the LR classifier is learnt in a \emph{centralized} environment. In particular, the poisoned data can increase the testing error rate of the LR classifier from 0.10 to 0.16 (60\% relative increase) in centralized learning. 
However, in federated learning,  the attacker can only inject the poisoned data to the compromised worker devices. We consider two scenarios on how the attacker distributes the poisoned data to the compromised worker devices:  

{\bf Single worker.} In this scenario, the attacker distributes the poisoned data on a single compromised worker device. 

{\bf Uniform distribution.} In this scenario, the attacker distributes the poisoned data to the compromised worker devices uniformly at random.

\alan{We consider the two scenarios because they represent two extremes for distributing data (concentrated or evenly distributed) and we expect one extreme to maximize attack effectiveness.} 
\neil{Table~\ref{data_poison_result_knownagg} compares BGA with our attacks. We observe that BGA has limited success at attacking Byzantine-robust aggregation rules, while our attacks can substantially increase the testing error rates. We note that if the federated learning uses the \emph{mean} aggregation rule BGA is still successful. For instance, when the mean aggregation rule is used, BGA can increase the testing error rate by 50\% when distributing the poisoned data to the compromised worker devices uniformly at random. However, when applying our attacks for trimmed mean to attack the mean aggregation rule, we can increase the testing error rates substantially more (see the last two cells in the second row of Table~\ref{data_poison_result_knownagg}). }

\begin{table}[!t]\renewcommand{\arraystretch}{1}
\centering
\caption{\neil{Testing error rates of back-gradient optimization based attacks (SingleWorker and Uniform) and our attacks (Partial and Full).}}
\centering
\addtolength{\tabcolsep}{-4pt}
\begin{tabular}{|c|c|c|c|c|c|} \hline 
{\small } & {\small NoAttack} & {\small SingleWorker} & {\small Uniform} &{\small  Partial} & {\small Full}  \\ \hline
{\small Mean}  & {\small 0.10} & {\small 0.11} & {\small 0.15} &{\small 0.54} & {\small 0.69}  \\ \hline
{\small Krum}  & {\small 0.23} & {\small 0.24} & {\small 0.25} &{\small 0.85} & {\small 0.89}  \\ \hline
{\small Trimmed mean}  & {\small 0.12} & {\small 0.12} & {\small 0.13} &{\small 0.27} & {\small 0.32}  \\  \hline
{\small Median}  & {\small 0.13} & {\small 0.13} & {\small 0.14} &{\small 0.19} & {\small 0.21}  \\ \hline 
\end{tabular}
\vspace{-3mm}
\label{data_poison_result_knownagg}
\end{table}


\vspace{-3mm}
\section{Defenses}
\label{sec:defense}

We generalize RONI~\cite{barreno2010security} and TRIM~\cite{Jagielski18}, which were designed to defend against data poisoning attacks, to defend against our local model poisoning attacks. Both generalized defenses remove the local models that are potentially malicious before computing the global model in each iteration of federated learning.  One generalized defense removes the local models that have large negative impact on the error rate of the global model (inspired by RONI that removes training examples that have large negative impact on the error rate of the model), while the other defense removes the local models that result in large loss (inspired by TRIM that removes the training examples that have large negative impact on the loss). In both defenses, we assume the master device has a small \emph{validation dataset}. Like existing aggregation rules such as Krum and trimmed mean, we assume the master device knows the upper bound $c$ of the number of compromised worker devices. \alan{We note that our defenses make the global model slower to learn and adapt to new data as that data may be identified as from potentially malicious local models.}

\myparatight{Error Rate based Rejection (ERR)}
In this defense, we compute the impact of each local model on the error rate for the validation dataset and remove the local models that have large negative impact on the error rate. Specifically, suppose we have an aggregation rule. For each local model, we use the aggregation rule to compute a global model $A$ when the local model is included and a global model $B$ when the local model is excluded. We compute the error rates of the global models $A$ and $B$ on the validation dataset, which we denote as $E_A$ and $E_B$, respectively. We define $E_A-E_B$ as the \emph{error rate impact} of a local model. A larger error rate impact indicates that the local model increases the error rate more significantly if we include the local model when updating the global model. We remove the $c$ local models that have the largest error rate impact, and we aggregate the remaining local models to obtain an updated global model.

\myparatight{Loss Function based Rejection (LFR)} In this defense, we remove local models based on their  impact on the loss instead of error rate for the validation dataset. Specifically, like the error rate based rejection, for each local model, we compute the global models $A$ and $B$. We compute the \alan{cross-entropy loss function values} of the models $A$ and $B$ on the validation dataset, which we denote as  $L_A$ and $L_B$, respectively. Moreover, we define $L_A-L_B$ as the \emph{loss impact} of the local model. Like the error rate based rejection, we remove the $c$ local models that have the largest loss impact, and we aggregate the remaining local models to update the global model. 

\neil{\myparatight{Union (i.e., ERR+LFR)} In this defense, we combine ERR and LFR. Specifically, we remove the local models that are removed by either ERR or LFR.}

\begin{table}[!t]\renewcommand{\arraystretch}{1}
\centering
\caption{\neil{Defense results.  The numbers are testing error rates. The columns ``Krum'' and ``Trimmed mean'' indicate  the attacker's assumed aggregation rule when performing attacks, while the rows indicate the actual aggregation rules and defenses. Partial knowledge attacks are considered.}}
\centering
\addtolength{\tabcolsep}{-2pt}
\begin{tabular}{|c|c|c|c|} \hline 
{\small } & {\small No attack} & {\small Krum} & {\small Trimmed mean}  \\ \hline
{\small Krum } & {\small 0.14} & {\small 0.72} & {\small 0.13}\\ \hline
{\small Krum + ERR} & {\small 0.14} & {\small 0.62} & {\small 0.13}\\ \hline
{\small Krum + LFR} & {\small 0.14} & {\small 0.58} & {\small 0.14}\\ \hline 
{\small Krum + Union} & {\small 0.14} & {\small 0.48} & {\small 0.14}\\ \hline \hline
{\small Trimmed mean } & {\small 0.12} & {\small 0.15} & {\small 0.23}\\ \hline
{\small Trimmed mean + ERR} & {\small 0.12} & {\small 0.17} & {\small 0.21}\\ \hline
{\small Trimmed mean + LFR} & {\small 0.12} & {\small 0.18} & {\small 0.12}\\ \hline 
{\small Trimmed mean + Union} & {\small 0.12} & {\small 0.18} & {\small 0.12}\\ \hline \hline 
{\small Median } & {\small 0.13} & {\small 0.17} & {\small 0.19}\\ \hline
{\small Median + ERR} & {\small 0.13} & {\small 0.21} & {\small 0.25}\\ \hline
{\small Median + LFR} & {\small 0.13} & {\small 0.20} & {\small 0.13}\\ \hline
{\small Median + Union} & {\small 0.13} & {\small 0.19} & {\small 0.14}\\ \hline  
\end{tabular}
\label{defense-result}
\end{table}

\myparatight{Defense results} \alan{Table~\ref{defense-result} shows the defense results of ERR, FLR, and Union, where partial knowledge attacks are considered. We use the default parameter setting discussed in Section~\ref{parametersetting}, e.g., 100 worker devices, 20\% of compromised worker devices, MNIST dataset, and LR classifier.    Moreover, we  sample 100 testing examples uniformly at random as the validation dataset. Each row of the table corresponds to a defense, e.g., Krum + ERR means that the master device uses ERR to remove the potentially malicious local models and uses Krum as the aggregation rule. Each column indicates the attacker's assumed aggregation rule when performing attacks, e.g., the column ``Krum'' corresponds to attacks that are based on Krum. We have several observations.}

First, LFR is comparable to ERR or much more effective than ERR, i.e., LFR achieves similar or much smaller testing error rates than ERR. For instance, Trimmed mean + ERR and Trimmed mean + LFR achieve similar testing error rates (0.17 vs. 0.18) when the attacker crafts the compromised local models based on Krum. However, Trimmed mean + LFR achieves a much smaller testing error rate than Trimmed mean + ERR (0.12 vs. 0.21),  when the attacker crafts the compromised local models based on trimmed mean. \neil{Second, Union is comparable to LFR in most cases, except one case (Krum + LFR vs. Krum and Krum + Union vs. Krum) where Union is  more effective.} 

\neil{Third, LFR and Union can effectively defend against our attacks in some cases. For instance, Trimmed mean + LFR (or Trimmed mean + Union) achieves the same testing error rate for both no attack and attack based on trimmed mean. However,  our attacks are still  effective in other cases even if LFR or Union is adopted. For instance, an attack, which crafts compromised local models based on Krum, 
 still effectively increases the error rate from 0.14 (no attack) to 0.58 (314\% relative increase) for Krum + LFR. Fourth, the testing error rate grows in some cases when a defense is deployed. This is because the defenses may remove benign local models, which increases the testing error rate of the global model.}


\vspace{-3mm}
\section{Related Work}
\label{related}

Security and privacy of federated/collaborative learning are much less explored, compared to centralized machine learning. Recent studies~\cite{Hitaj17,Melis19,Nasr19}  explored privacy risks in federated learning, which are orthogonal to our study. 
 
\alan{\myparatight{Poisoning attacks} Poisoning attacks aim to compromise the integrity of the training phase of a machine learning system~\cite{barreno2006can}. The training phase  consists of two components, i.e., training dataset collection and learning process. 
Most existing poisoning attacks compromise the training dataset collection component, e.g., inject malicious data into the training dataset. These attacks are also known as \emph{data poisoning attacks})~\cite{rubinstein2009antidote,biggio2012poisoning,xiao2015feature,Jagielski18,poisoningattackRecSys16,YangRecSys17,Nelson08poisoningattackSpamfilter,munoz2017towards,Suciu18,Gu17,Chen17,Bagdasaryan18,shafahi2018poison,Fang18,Wang19}. Different from data poisoning attacks, our local model poisoning attacks compromise the learning process.

Depending on the goal of a poisoning attack, we can classify poisoning attacks into two categories, i.e., \emph{untargeted poisoning attacks}~\cite{rubinstein2009antidote,biggio2012poisoning,xiao2015feature,Jagielski18,poisoningattackRecSys16,YangRecSys17} and \emph{targeted poisoning attacks}~\cite{Nelson08poisoningattackSpamfilter,Suciu18,Liu18,Gu17,Chen17,Bagdasaryan18,shafahi2018poison,Bhagoji19}.  Untargeted poisoning attacks aim to make the learnt model have a high testing error indiscriminately for testing examples, which eventually result in a denial-of-service attack.  
In targeted poisoning attacks,  the learnt model produces  attacker-desired predictions for particular testing examples, e.g., predicting spams as non-spams and predicting attacker-desired labels for testing examples with a particular trojan trigger (these attacks are also known as \emph{backdoor/trojan attacks}~\cite{Gu17}). However, the testing error for other  testing examples is unaffected. 
Our local model poisoning attacks are untargeted poisoning attacks. Different from existing untargeted poisoning attacks that focus on centralized machine learning, our attacks are optimized for Byzantine-robust federated learning. We note that Xie et al.~\cite{Xie19} proposed inner product manipulation based untargeted poisoning attacks to Byzantine-robust federated learning including Krum and median, which is concurrent to our work.

\myparatight{Defenses}  Existing defenses were mainly designed for data poisoning attacks to centralized machine learning. They essentially aim to detect the injected malicious data in the training dataset. One category of defenses~\cite{Cretu08,barreno2010security,Suciu18,Tran18} detects malicious data based on their (negative) impact on the performance of the learnt model. 
For instance, 
Barreno et al.~\cite{barreno2010security} proposed \emph{Reject on Negative Impact (RONI)}, which measures the impact of each training example on the performance of the learnt model and removes the training examples that have large negative impact. Suciu et al.~\cite{Suciu18} proposed a variant of RONI (called tRONI) for targeted poisoning attacks. In particular, tRONI measures the impact of a training example on only the target classification and excludes training examples that have large impact. 

Another category of defenses~\cite{Feng14,Liu17AiSec,Jagielski18,Steinhardt17} proposed new loss functions, optimizing which obtains model parameters and detects the injected malicious data simultaneously. For instance, 
Jagielski et al.~\cite{Jagielski18} proposed TRIM, which aims to jointly find a subset of training dataset with a given size and model parameters that minimize the loss function. The training examples that are not in the selected subset are treated as malicious data. 
These defenses are not directly applicable for our local model poisoning attacks because our attacks do not inject malicious data into the training dataset.

For federated learning, the machine learning community recently proposed several aggregation rules (e.g., Krum~\cite{Blanchard17}, Bulyan~\cite{Mhamdi18}, trimmed mean~\cite{Yin18}, median~\cite{Yin18}, and others~\cite{ChenPOMACS17}) that were claimed to be robust against Byzantine failures of certain worker devices. Our work shows that these defenses are not effective in practice against our optimized local model poisoning attacks that carefully craft local models on the compromised worker devices. Fung et al.~\cite{Fung18} proposed to compute weight for each worker device according to historical local models and take the weighted average of the local models to update the global model. However, their method can only defend against label flipping attacks, which can already be defended by existing Byzantine-robust aggregation rules. We propose ERR and LFR, which are respectively generalized from RONI and TRIM, to defend against our local model poisoning attacks. We find that these defenses are not effective enough in some scenarios, highlighting the needs of new defenses against our attacks. 

\myparatight{Other security and privacy threats to machine learning} Adversarial examples~\cite{barreno2006can,szegedy2013intriguing} aim to make a machine learning system predict labels as an attacker desires via adding carefully crafted noise to normal testing examples in the testing phase. Various methods (e.g.,~\cite{szegedy2013intriguing,goodfellow2014explaining,Papernot16,CarliniSP17,laskov2014practical,sharif2016accessorize,liu2016delving,PracticalBlackBox17,Athalye18}) were proposed to generate adversarial examples, and many defenses (e.g.,~\cite{goodfellow2014explaining,madry2017towards,Papernot16Distillation,detection2,detection1,xu2017feature,region}) were explored to mitigate them. Different from poisoning attacks, adversarial examples compromise the testing phase of machine learning. Both poisoning attacks and adversarial examples compromise the integrity of machine learning. An attacker could also compromise the confidentiality of machine learning. Specifically, an attacker could compromise the confidentiality of users' private training or testing data via various attacks such as \emph{model inversion attacks}~\cite{fredrikson2014privacy,fredrikson2015model}, \emph{membership inference attacks}~\cite{membershipInfer,membershipLocation,Melis19}, and \emph{property inference attacks}~\cite{Ateniese15,Ganju18}. Moreover, an attacker could also  compromise the confidentiality/intellectual property of a model provider via stealing its model parameters and hyperparameters~\cite{liang2016cracking,tramer2016stealing,WangHyper18}. 
}


\vspace{-3mm}
\section{Conclusion, Limitations, and Future Work}

We demonstrate that the federated learning methods, which the machine learning community claimed to be robust against Byzantine failures of some worker devices, are vulnerable to our local model poisoning attacks that manipulate the local models sent from the compromised worker devices to the master device during the learning process. In particular, to increase the error rates of the learnt global models, an attacker can craft the local models on the compromised worker devices such that the aggregated global model deviates the most towards the inverse of the direction along which the global model would change when there are no attacks. Moreover, finding such crafted local models can be formulated as optimization problems. We can generalize existing defenses for data poisoning attacks to defend against our local model poisoning attacks. Such generalized defenses are effective in some cases but are not effective enough in other cases. Our results highlight that we need new defenses to defend against our local model poisoning attacks. 

Our work is limited to untargeted poisoning attacks. It would be interesting to study targeted poisoning attacks to federated learning. Moreover, it is valuable future work to design new defenses against our local model poisoning attacks, e.g., new methods to detect compromised local models and new adversarially robust aggregation rules. 


\vspace{-3mm}
\section{Acknowledgements}

We thank the anonymous reviewers and our shepherd Nikita Borisov for constructive reviews and comments. This work was supported by NSF grant No.1937786.

\vspace{-3mm}


\bibliographystyle{plain}
\bibliography{refs,refs2}


\appendix
 
\begin{table}[!t]\renewcommand{\arraystretch}{0.9}
\caption{(a) The DNN architecture (input layer is not shown) used for MNIST and Fashion MNIST. (b) Testing error rates when applying attacks for Krum  to attack Bulyan.}
\vspace{-2mm}
\addtolength{\tabcolsep}{-4pt}
\subfloat[]{
\begin{tabular}{|c|c|} \hline 
{\small Layer Type} & {\small Size} \\ \hline
{\small Convolution + ReLU} & {\small$3\times3\times30$}\\ \hline
{\small Max Pooling} & {\small$2\times2$}\\ \hline
{\small Convolution + ReLU} & {\small$3\times3\times50$}\\ \hline
{\small Max Pooling} & {\small$2\times2$}\\ \hline
{\small Fully Connected + ReLU} & {\small 200}\\ \hline
{\small Softmax} & {\small 10 / 8}\\ \hline
\end{tabular}
\label{DNN-architecture}
}
\subfloat[]{
\begin{tabular}{|c|c|} \hline 
{\small } & {\small  Bulyan} \\ \hline
{\small No attack} & {\small 0.14} \\ \hline
{\small Partial Knowledge} & {\small 0.36} \\ \hline
{\small Full Knowledge} & {\small 0.38}\\ \hline
\end{tabular}
\label{bulyanResult}}
\vspace{-4mm}
\end{table}

\begin{table}[!t]\renewcommand{\arraystretch}{0.9}
\centering
\caption{\neil{Testing error rates of our attacks based on the deviation goal and directed deviation goal.}}
\vspace{-1mm}
\centering
\begin{tabular}{|c|c|c|c|} \hline 
{\small } & {\small Krum} & {\small Trimmed mean} &{\small Median} \\ \hline
{\small Deviation goal} & {\small 0.87} & {\small 0.10} & {\small 0.12}\\ \hline
{\small Directed deviation goal} & {\small 0.80} & {\small 0.52} & {\small 0.29}\\ \hline
\end{tabular}
\vspace{-3mm}
\label{deviationResult}
\end{table}

\vspace{-3mm}
\section{Attacking Bulyan}
\label{attackBulyan}

Bulyan is based on Krum. We apply our attacks for Krum to attack Bulyan. Table~\ref{bulyanResult} shows results of attacking Bulyan. The dataset is MNIST, the classifier is logistic regression, $m=100$, $c=20$, $\theta=m-2c$ (Bulyan selects $\theta$ local models using Krum), and $\gamma=\theta-2c$ (Bulyan takes the mean of $\gamma$ parameters). 
Our results show that our attacks to Krum can transfer to Bulyan. Specifically, our partial knowledge attack increases the error rate by around 150\%, while our full knowledge attack  increases the error rate by 165\%.

\section{Deviation Goal}
\label{goals}

The deviation goal  is to craft local models $\mathbf{w}_1', \mathbf{w}_2', \cdots, \mathbf{w}_c'$ for the compromised worker devices via solving the following optimization problem in each iteration:
\begin{align}
\label{problem0}
&\max_{\mathbf{w}_1', \cdots, \mathbf{w}_c'} ||\mathbf{w} - \mathbf{w}'||_1, \nonumber \\
\text{subject to } & \mathbf{w}=\mathcal{A}(\mathbf{w}_1, \cdots, \mathbf{w}_c, \mathbf{w}_{c+1}, \cdots, \mathbf{w}_m), \nonumber \\
& \mathbf{w}'=\mathcal{A}(\mathbf{w}_1', \cdots, \mathbf{w}_c', \mathbf{w}_{c+1}, \cdots, \mathbf{w}_m),
\end{align} 
where $||\cdot||_1$ is $L_1$ norm. 
We can adapt our attacks based on the directed deviation goal to the deviation goal. 
For simplicity, we focus on the full knowledge scenario.

\myparatight{Krum} Similar to the directed deviation goal, we make two approximations, i.e., $\mathbf{w}_{1}' = {\mathbf{w}}_{Re} - \lambda  $ and the $c$ compromised local models are the same. Then, we formulate an optimization problem similar to Equation~\ref{problem2}, except that  $\mathbf{w}_1'= {\mathbf{w}}_{Re} - \lambda \mathbf{s}$ is changed to $\mathbf{w}_1'= {\mathbf{w}}_{Re} - \lambda$. 
Like Theorem~\ref{theoremLambda}, we can derive an upper bound of $\lambda$, given which we use binary search to solve $\lambda$. After solving $\lambda$, we obtain $\mathbf{w}_{1}'$. Then, we randomly sample $c-1$ vectors whose Euclidean distances to $\mathbf{w}_{1}'$ are smaller than $\epsilon$ as the other $c-1$ compromised local models. 

\myparatight{Trimmed mean} Theoretically, we can show that the following attack can maximize the deviation of the global model:  we use any $c$ numbers that are larger than $w_{max,j}$ or smaller than $w_{min,j}$, depending on which one makes the deviation larger, as the $j$th local model parameters on the $c$ compromised worker devices. Like the directed deviation goal, when implementing the attack, we randomly sample the $c$ numbers in the interval [$w_{max,j}, b\cdot w_{max,j}$] (when $w_{max,j} > 0$) or [$w_{max,j}, w_{max,j}/b$] (when $w_{max,j} \leq 0$), or in the interval [$w_{min,j}/b, w_{min,j}$] (when $w_{min,j} > 0$) or [$b\cdot w_{min,j}, w_{min,j}$] (when $w_{min,j} \leq 0$), depending on which one makes the deviation larger. 

\noindent
{\bf Median:} We apply the attack for trimmed mean to median. 

\noindent
{\bf Experimental results:} Table~\ref{deviationResult} empirically compares the deviation goal and directed deviation goal, where  MNIST and  LR classifier are used. 
For Krum, both  goals achieve high testing error rates. However, for trimmed mean and median,  the directed deviation goal achieves significantly higher testing error rates than the deviation goal.

\vspace{-2mm}
\section{Proof of Theorem~\ref{theoremLambda}}
\label{appendix:proof}

We denote by ${\Gamma}_{\mathbf{w}}^{a}$ the set of $a$ local models among the crafted $c$ compromised local models and $m-c$ benign local models that are the closest to the local model $\mathbf{w}$ with respect to Euclidean distance. 
Moreover, we denote by ${\tilde{\Gamma}_{\mathbf{w}}^{a}}$  the set of 
$a$ benign local models that are the closest to $\mathbf{w}$ with respect to Euclidean distance. Since $\mathbf{w}_1'$ is chosen by Krum, we have the following:
\begin{align}
\sum_{l\in {\Gamma}_{\mathbf{w}_1'}^{m-c-2}} D^2(\mathbf{w}_l,\mathbf{w}_1') \le \min_{c+1\le i\le m}{\sum_{l\in {\Gamma}_{\mathbf{w}_i}^{m-c-2}}} D^2(\mathbf{w}_l,\mathbf{w}_i),
\end{align}
where $D(\cdot,\cdot)$ represents Euclidean distance. 
The distance between $\mathbf{w}_1'$ and the other $c-1$ compromised local models is 0, since we assume they are the same in the optimization problem in Equation~\ref{problem2} when finding $\mathbf{w}_1'$. Therefore, we have:
\begin{align}
\sum_{l\in {\tilde{\Gamma}_{\mathbf{w}_1'}^{m-2c-1}}} D^2(\mathbf{w}_l, \mathbf{w}_1') \le \min_{c+1\le i\le m}{\sum_{l\in {\Gamma}_{\mathbf{w}_i}^{m-c-2}}} D^2(\mathbf{w}_l,\mathbf{w}_i).
\end{align}

According to the triangle inequality, we have $D^2(\mathbf{w}_l, \mathbf{w}_1') \ge [D(\mathbf{w}_1', \mathbf{w}_{Re}) - D(\mathbf{w}_l,\mathbf{w}_{Re})]^2$. 
Since $D(\mathbf{w}_1',\mathbf{w}_{Re})=\Vert{\lambda\cdot \mathbf{s}}\Vert_2=\sqrt{d}\cdot\lambda$, we have:
\begin{align}
	\sum_{l\in {\tilde{\Gamma}_{\mathbf{w}_1'}^{m-2c-1}}} [\sqrt{d}\cdot\lambda - D(\mathbf{w}_l, \mathbf{w}_{Re})]^2 \le \min_{c+1\le i\le m}{\sum_{l\in {\Gamma}_{\mathbf{w}_i}^{m-c-2}}} D^2(\mathbf{w}_l,\mathbf{w}_i),\nonumber
\end{align}
which gives the following necessary condition:
\begin{align}
	\lambda \le & \sqrt{\frac{1}{(m-2c-1)d}  \cdot \min_{c+1\le i\le m}{\sum_{l\in {\tilde{\Gamma}_{\mathbf{w}_i}^{m-c-2}}}} D^2(\mathbf{w}_l,\mathbf{w}_i)} \nonumber\\
			&  +  \frac{1}{\sqrt{d}}\cdot \max_{c+1\le i\le m}{D(\mathbf{w}_i,\mathbf{w}_{Re})}.
\end{align}
The bound only depends on the before-attack local models. 


\end{document}